\documentclass[final, twoside]{IEEEtran} 
\usepackage{float}
\usepackage{amssymb, amsmath, amsthm, bm}

\makeatletter
\newcommand{\tpmod}[1]{{\@displayfalse\pmod{#1}}}
\makeatother

\usepackage{graphicx, color}
\graphicspath{{figures-pdf/}}
\usepackage{cite}
\usepackage{url}
\usepackage{diagbox}
\usepackage[aboveskip=1pt]{subcaption}
\usepackage{threeparttable}
\usepackage[mathlines]{lineno}
\usepackage[linesnumbered, ruled]{algorithm2e}

\usepackage{mleftright}
\mleftright
\usepackage{booktabs}
\usepackage{algpseudocode}
\usepackage{multirow, bigstrut}
\usepackage{tabularx}
\usepackage{arydshln}
\usepackage{empheq}
\usepackage{datetime}
\usepackage{makecell}

\newlength\OneImW
\newlength\BigOneImW
\newlength\ThreeImW
\newlength\twofigwidth
\newlength\vfigskip
\newlength\figsep
\newtheorem{Proposition}{Proposition}
\newtheorem{Corollary}{Corollary}

\newtheorem{Property}{Property}

\newtheorem{Definition}{Definition}

\newcommand{\X}[0]{{\bf X}}
\newcommand{\Y}[0]{{\bf Y}}
\newcommand{\s}[0]{{\bf S}}
\newcommand{\x}[0]{{\bf x}}

\newcommand{\BC}[0]{\overline{\mathbb{Y}}}
\newcommand{\CEIL}[1]{\left\lceil{#1}\right\rceil}
\usepackage[bookmarks=false]{hyperref}
\hypersetup{
 linktocpage=true, pdfborderstyle={/S/S/W 1}, hyperindex=true, bookmarks=true, bookmarksopen=true, bookmarksnumbered=true,
}

\usepackage{textcomp}

\begin{document}

\title{The Graph Structure of Baker's maps Implemented on a Computer}
\author{Chengqing Li, Kai Tan
\thanks{This work was supported by the National Natural Science Foundation of China (no.~92267102, 61772447).}

\thanks{C. Li and K. Tan are with the School of Computer Science, Xiangtan University, Xiangtan 411105, Hunan, China.}
}

\markboth{IEEE Transactions}{Li \MakeLowercase{et al.}}
\IEEEpubid{\begin{minipage}{\textwidth}\ \\[12pt] \centering
			1520-9210 \copyright 2020 IEEE. Personal use is permitted, but republication/redistribution requires IEEE permission.\\
			See https://www.ieee.org/publications/rights/index.html for more information. \\
\today
\end{minipage}
}

\maketitle

\begin{abstract}
The complex dynamics of baker's map and its variants in an infinite-precision mathematical domain have been extensively analyzed in the past five decades. However, their real structure implemented in a finite-precision computer remains unclear.
This paper gives an explicit formulation for the quantized baker's map and its extension into higher dimensions.
Our study reveals certain properties, such as the in-degree distribution in the state-mapping network approaching a constant with increasing precision, and a consistent maximum in-degree across various levels of fixed-point arithmetic precision. We also observe a fractal pattern in baker's map functional graph as precision increases, characterized by fractal dimensions.
We then thoroughly examine the structural nuances of functional graphs created by the higher-dimensional baker's map (HDBM) in both fixed-point and floating-point arithmetic domains. An interesting aspect of our study is 
the use of interval arithmetic to establish a relationship between the HDBM's functional graphs across these two computational domains.
A particularly intriguing discovery is the emergence of a semi-fractal pattern within the functional graph of a specific baker's map variant, observed as the precision is incrementally increased. 
The insights gained from our research offer a foundational understanding for the dynamic analysis and application of baker's map and its variants in various
domains.
\end{abstract}

\section{Introduction}

\IEEEPARstart{I}{n} the past two decades, chaotic maps have gained significant attention in various fields, including 
secure communication \cite{cqli:Diode:TCASI19, Mehallel:Enhancement:2021,cqli:HNN:TCASI2024},
quantum chaos \cite{Arul:PRE:2019, Clauss:Universal:2021, Mudute-Ndumbe:non-hermitianpt:2020, Hou:Quantum:2020},
image encryption \cite{Fridrich:ChaoticImageEncryption:IJBC98, Alvarez:Baker:PLA2006}, and
pseudorandom number generator \cite{Chen:logistic:TCASII10,garcia2018chaos:TIM18,wangqx:higher:TCASI21,Galias:Henon:TCASI22, chen:PRNS:TC22, licq:logistic:IJBC2023}.
baker's map, renowned for its simulation of dough kneading, stands out as a prominent example \cite{chen:Baker:IJBCC04}.
Due to the seemingly strong cryptographical potential of complex dynamics of chaotic maps, it is used as a source to produce pseudorandom number sequence for controlling basic encryption operations and their combination
\cite{Machado:Cryptography:2004}. In \cite{Mehallel:Enhancement:2021}, the discrete baker's map is used to design chaotic interleaving algorithms,
which 
rearrange the positions of image pixels and mitigate transmission errors.
As a classical model of chaotic resonance, baker's map has been analyzed through the statistics of its single resonance state \cite{Clauss:Universal:2021}.
Notably, in 1989, the quantum version of baker’s map was experimentally implemented on a three-bit quantum information processor using nuclear magnetic resonance \cite{PhysRevLett.89.157902}.
The release of IBM's first circuit-based commercial quantum computer in January 2019 marked a significant advancement. It further encouraged the use of discrete baker's mapping for the efficient switching of quantum bits with minimal circuit complexity.

Baker's map is a bijection from the unit square $[0, 1)\times [0, 1)$ onto itself and defined by
\begin{equation}
\label{eq:Baker}
\mathrm{B}(x, y)=
 \begin{cases}
 (2x, y/2) & \text{if } 0 \le x <1/2; \\
 (2x-1, y/2+1/2) & \text{if } 1/2\le x < 1.
 \end{cases}
\end{equation}
As shown in Fig.~\ref{Fig:GBxy}, baker's map is generalized by dividing a square into $k$ strips, $[F_i, F_{i+1})\times[0, 1)$, and mapping each vertical strip to a horizontal one $[0, 1)\times[F_i, F_{i+1})$.
Mathematically, the generalized baker's map can be presented as
\begin{equation}
\label{eq:Gbaker}
\mathrm{B}(x, y)=\left( \frac{1}{p_n}(x-F_n), p_n\cdot y+F_n \right),
\end{equation}
where $1\leq n \leq k$ and $(x, y)\in [F_n, F_{n+1})\times [0, 1)$. 

The fundamental properties of baker's map~\eqref{eq:Baker} are disclosed in Arnold's classic monograph \cite{Arnold:WAB:1968}.
In \cite{Schack:Baker:PRL92}, much more on dynamics of the continuous baker's map demonstrated under various conditions were extensively investigated, e.g.
hypersensitivity to perturbation. 
There are other ways to extend baker's map \cite{chen:Baker:IJBCC04,YaobinMao:CSF2004}. 
Because of the subtle similarity between the
properties of the generalized baker's map~\eqref{eq:Gbaker}
and that of a cryptosystem, it is further extended by Bernoulli permutation to construct a symmetric product cipher and a pseudo-random number generator \cite{Scharinger:Baker:1995}.
In \cite{Ozturk:Baker:TCASI2019}, it was extended to a higher-dimensional version called higher-dimensional baker's map (HDBM), exhibiting Devaney chaos in every dimension.

\begin{figure}[!htb]
\centering
 \begin{minipage}{0.6\ThreeImW}
 \includegraphics[width=0.6\ThreeImW]{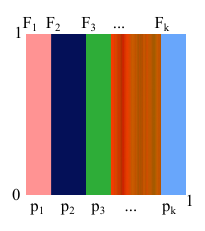}
 \end{minipage}
 \begin{minipage}{0.3\ThreeImW}
 \includegraphics[width=0.3\ThreeImW]{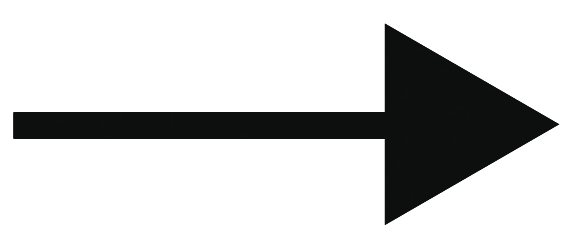}
 \end{minipage}
 \begin{minipage}{0.6\ThreeImW}
 \includegraphics[width=0.6\ThreeImW]{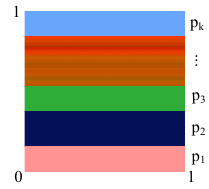}
 \end{minipage}
\caption{Demonstration of the generalized baker's map~(\ref{eq:Gbaker}).}
\label{Fig:GBxy}
\end{figure}

\IEEEpubidadjcol 

When a chaotic map is implemented on a digital device or simulated via a software suite, its dynamical properties may degrade or even diminish \cite{Persohn:CSF:2012, Boghosian:Pathology:ATS19}.
The randomness of the sequence generated by iterating a chaotic map on a digital device is much lower than expected.
In 2006, the insecurity problem of an encryption scheme caused by the dynamical degradation of baker's map implemented with computer number format Binary32 or Binary64 was studied in \cite{Alvarez:Baker:PLA2006}.
Then, some methods are designed for implementing chaotic systems on digital platforms, such as counteracting dynamical degradation of chaotic systems in an integer domain
with true random external stimulus \cite{wangqx:Theoretical:TCASI16} and filling the least-significant bits (LSB) of each state with some random bits \cite{Ozturk:Baker:TCASI2019, OZTURK2018395}.
Recently, we found that a functional graph, also named state-mapping network (SMN), can serve as a special perspective to observe
the internal structure of a chaotic map implemented in a finite-precision arithmetic domain, especially its changing rules with cumulative increment of the arithmetic precision. Some new intrinsic properties of two 1-D chaotic maps, Logistic map \cite{Li:Network:ISCAS2019}, Tent map \cite{cqli:network:TCASI2019},
and a 2-D Cat map \cite{cqli:Cat:TC22}, were unveiled. In addition, the functional graph composed by iterating a Chebyshev polynomial from every possible initial
state demonstrates very strong regularity with the increase of the fixed-point arithmetic precision \cite{Li:Network:CCC2021}.
Some specific floating-point formats have been adopted for some real applications,
such as parallel computation
\cite{Nathalie:NumericalRe:TC2014}.
Fortunately, a chaotic map implemented in any floating-point arithmetic domain can be studied via the functional graph. Moreover, the properties demonstrated over different domains are similar to different extents. So, this work adopts the IEEE 754 standard floating-point format as the implementation environment.

As for a 2-D chaotic map, the key study object becomes the complex relations between node $(i, j)$ in a fixed-point arithmetic domain $\mathbb{Z}_{2^e}$ and
four nodes $(2i, 2j)$, $(2i+1, 2j)$, $(2i, 2j+1)$ and $(2i+1, 2j+1)$ in $\mathbb{Z}_{2^{e+1}}$ \cite{cqli:Cat:TC22}.
For a higher-dimension chaotic map, it just needs to represent the corresponding nodes by a high-dimensional vector.
This paper proved that the maximum in-degree of the functional graph of baker's map implemented in any fixed-precision domain remains unchanged.
Especially, the structure of the functional graph constructed by 2-D baker's map with binary parameters demonstrates strong regular patterns.
Then, the dynamical properties of HDBM (the higher-dimension version of baker's map) in a digital domain are disclosed by studying its functional graph.
Finally, the relation between functional graphs of HDBM implemented in a fixed-point arithmetic domain and the corresponding floating-point arithmetic domain is revealed.

The rest of the paper is organized as follows. Section~\ref{sec:graph} analyzes the graph structure of the generalized baker's map in a fixed-point arithmetic domain. Section~\ref{sec:HDBM} discloses the property of the functional graph composed by a higher-dimensional baker's map in the fixed-point arithmetic domain and floating-point arithmetic domain and the relation between functional graphs in the two domains.
The last section concludes the paper.

\section{The graph structure of the generalized baker's map in a fixed-point arithmetic domain}
\label{sec:graph}

The generalized baker's map $\mathrm{B}(x, y)$: $[0, 1)^2 \rightarrow [0, 1)^2$ is a crucial example in dynamical systems, particularly in chaos theory and statistical mechanics. It meets Devaney's criteria for chaos, exhibiting sensitivity to initial conditions, topological mixing, and dense periodic orbits. As a result, it features unstable periodic orbits (UPOs) due to its chaotic nature. Although the map is deterministic and periodic, the stretching and folding process results in periodic orbits that are highly sensitive to initial conditions. Moreover, as a mixing and ergodic system, the trajectories of almost all points cover the entire area $[0, 1)^2$, leading to a uniform distribution over sufficiently long iterations.
A fundamental aspect of these dynamical properties is that the Lebesgue measure on $[0, 1)^2$ is invariant under the generalized baker's map, which acts as a measure-preserving diffeomorphism on the torus $\{(x,y)\bmod 1\}$. 
This invariance implies that the distribution of points in the phase space remains consistent over time, facilitating the analysis of long-term behavior and the emergence of chaotic dynamics. However, in the digital domain, both the domain and range become discrete sets. The measure of a discrete set is typically defined using counting measure, which implies that the map loses its measure-preserving property.

Assume the generalized baker's map is implemented in a digital domain with fixed-point precision $e$.
So, its domain and range are both discrete set
$\{(x_1, x_2) \mid x_1=\frac{X_1}{2^e}, x_2=\frac{X_2}{2^e}, X_i\in\mathbb{Z}_{2^e}\}$.
Its associate functional graph $\mathbb{F}_e$ can be built by the following way:
every possible state is viewed as a unique node;
node $\X=(X_1, X_2)$ is directly linked to $\Y=(Y_1, Y_2)$ if and only if $\X$ is the preimage of $\Y$ under the generalized baker's map,
namely $\Y=\mathrm{R}\left(2^e \cdot \mathrm{B}\left(\frac{X_1}{2^e}, \frac{X_2}{2^e}\right)\right)$,
where 
$\mathrm{R}(\cdot)$ is a quantization function.
Figure~\ref{fig:perioddistributionof212} depicts $\mathbb{F}_e$ corresponding to the generalized baker's map with $(p_1, p_2, p_3) = (\frac{1}{2}, \frac{1}{2}, \frac{1}{3})$ when increasing $e$ from two to four. In addition, some important notations  referenced multiple times are defined in Table~\ref{tab:list:symbol} for enhanced readability.

\setlength{\tabcolsep}{3pt}
\begin{table}[!htb]
\centering
\caption{Nomenclature}
\begin{tabular}{c|l}
\hline
\makecell[c]{\bf Symbol} & \makecell[c]{\bf Definition}  \\ \hline
$\mathbb{F}_{e}$ & \makecell[l]{The functional graph of the generalized Baker's map which is \\ implemented in a digital domain with fixed-point precision $e$.} \\\hline
$\s$ & The set of all pre-images of node $\Y$.  \\ \hline
$\s_n$ & \makecell[l]{The set of all pre-images of node $\Y$ whose x-coordinate falls \\ within the interval $[E_n, E_{n+1})$.} \\ \hline
$X_{2, \inf}$ & The infimum of the y-coordinate of all nodes in $\s_n$. 
 \\ \hline
$X_{2, \sup}$ & The supremum of the y-coordinate of all nodes in $\s_n$. \\ \hline
$|I|$ & The cardinality of the interval $I$ \\ \hline
$|\s|$ & The cardinality of set $\s$\\ \hline
\end{tabular}
\label{tab:list:symbol}
\end{table}

To facilitate the following discussion, the quantized generalized baker's map is expressed as
$\Y=\mathrm{B}_e(\X)$,
where
\begin{equation}
\label{eq:GBakerp}
\mathrm{B}_e(\X)= \left(\mathrm{R}((X_1-E_n)/p_n), \mathrm{R}\big(p_n\cdot X_2+E_n \big) \right),
\end{equation}
$E_n=2^e\cdot F_n$ and $X_1\in [E_n, E_{n+1})$ in $\mathbb{F}_e$. 
The last quantization operation can be presented as
$\Y=\mathrm{R}(\bf Y')$,
where
\begin{equation}
\label{eq:Y'}
 \Y'=2^e \cdot \mathrm{B}\left(\frac{X_1}{2^e}, \frac{X_2}{2^e} \right).
\end{equation}
Note that $p_n$ is well-determined in the digital domain with a given precision $e$: 
\begin{equation}
\label{eq:diffE_n}
E_{n+1}-E_n=2^e\cdot (F_{n+1} - F_n)=2^e\cdot p_n.
\end{equation}
As the found rules are similar for different quantization functions, we adopt
the floor function as $\mathrm{R}(\cdot)$ throughout this paper unless otherwise specified. 

From Eq.~\eqref{eq:GBakerp} and Eq.~\eqref{eq:Y'}, one can see that there may exist multiple preimages for a given node.
Thus, let $\s$ be the set of all preimages of a node $\Y$ under the quantized generalized baker's map, namely
\begin{equation*}
\label{eq:SY}
\s=\{ \X\ \vert\ \mathrm{B}_e(\X)=\Y \}.
\end{equation*}
In particular, let $\s_n$ be the set of all preimages of a non-leaf node whose first coordinate belongs to the $n$-th interval $[E_n, E_{n+1})$,
namely
\begin{equation}
\label{eq:SYs}
\s_n=\{\X\ \vert\ \mathrm{B}_e(\X)=\Y, \X \in [E_n, E_{n+1})\times [0, 2^e)\}.
\end{equation}
Based on Properties~\ref{le:X1Y1i} and~\ref{prop:X_2} on two coordinates of node $\X$ in $\s_n$,
Proposition~\ref{prop:2n*d} reveals the value range of the in-degree of any non-leaf node in $\mathbb{F}_e$ of baker's map~(\ref{eq:Gbaker}), which is determined by its control parameters.

\begin{figure}[!htb]
\centering
\begin{minipage}[t]{0.6\twofigwidth}
\centering
		\raisebox{0.8em}{
	 \includegraphics[width=0.6\twofigwidth]{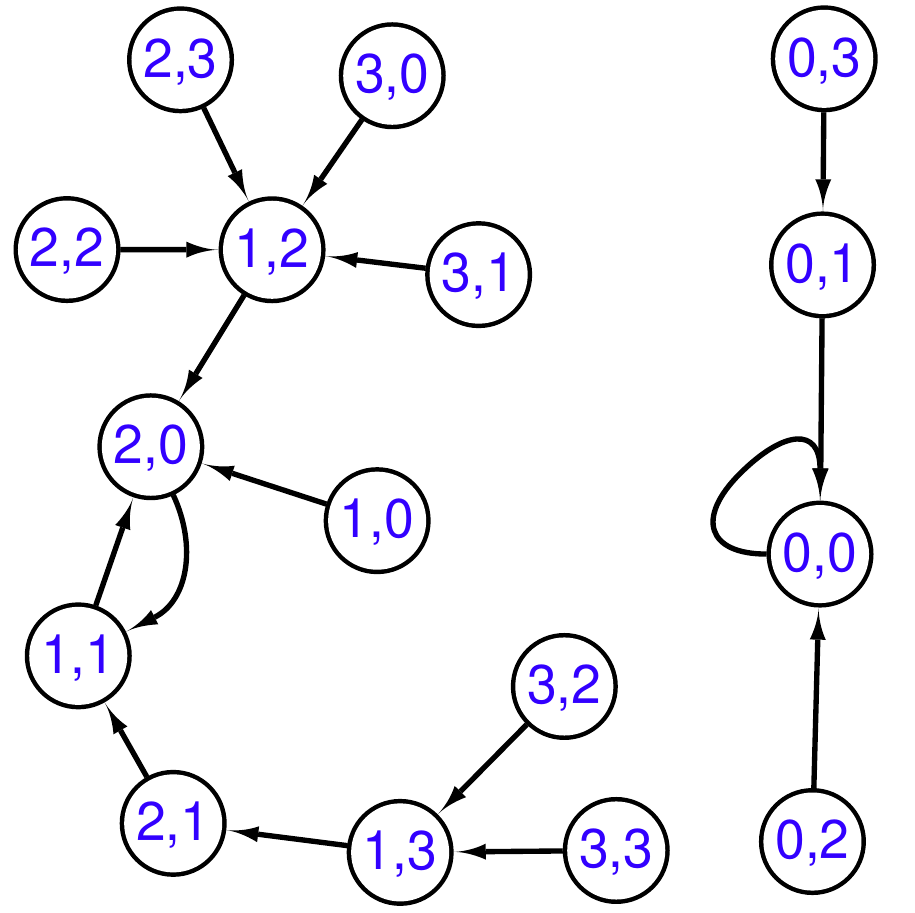}}
a)
	\end{minipage}
\hspace{2em}
 \begin{minipage}[t]{0.8\twofigwidth}
 \centering
 \includegraphics[width=0.8\twofigwidth]{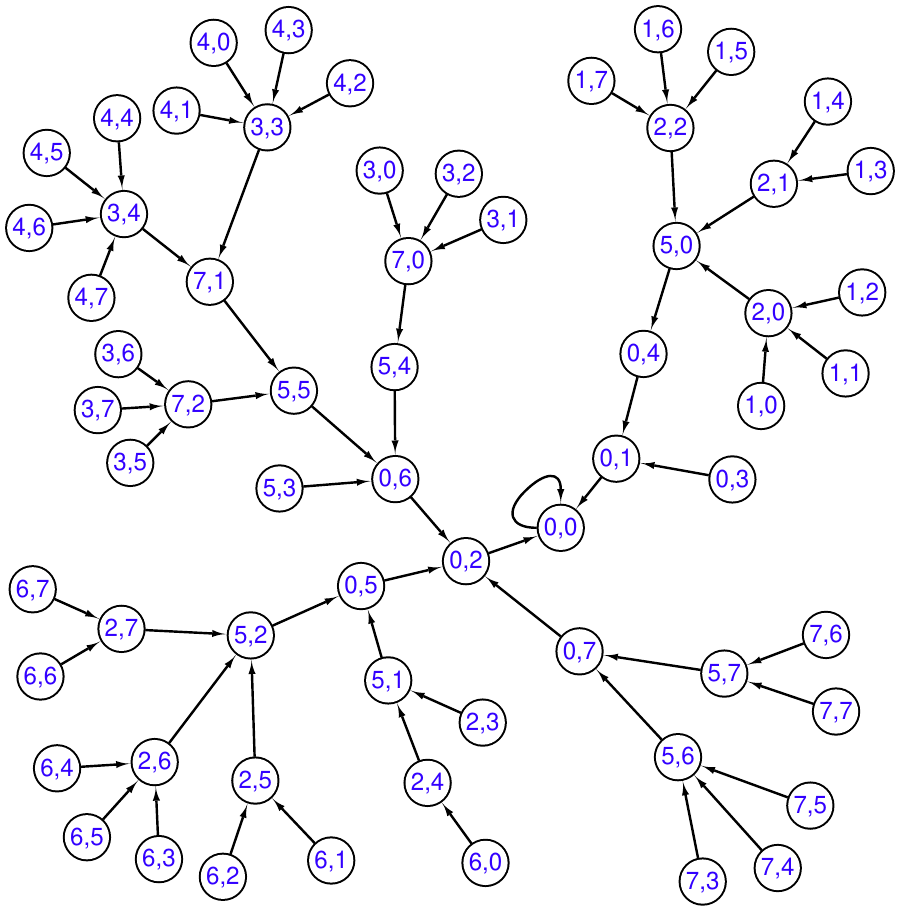}	
 b)
 \end{minipage}\vspace{0.5em}
\begin{minipage}[t]{1.4\twofigwidth}
\centering \includegraphics[width=1.4\twofigwidth]{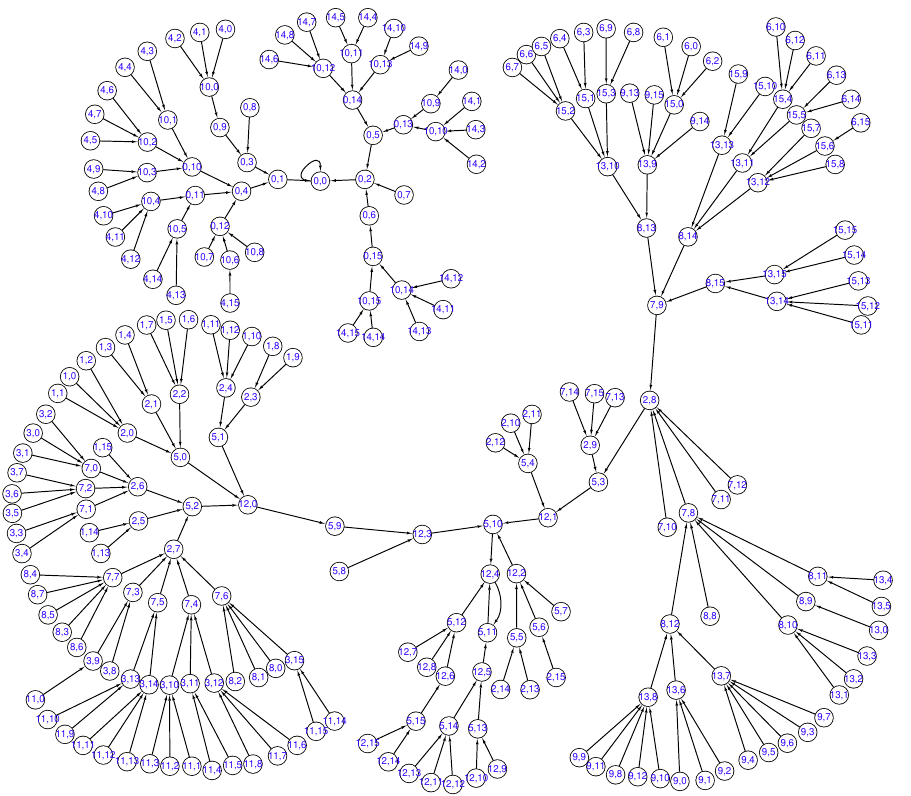}
c)
\end{minipage}
\caption{Functional graphs of the generalized baker's map~(\ref{eq:Gbaker}) with
$(p_1, p_2, p_3)=(\frac{2}{5}, \frac{1}{5}, \frac{2}{5})$ implemented with
precision $e$: a) $e=2$; b) $e=3$; c) $e=4$.}
\label{fig:perioddistributionof212}
\end{figure}

\begin{Property}
\label{le:X1Y1i}
For any node $\Y\in [0, 2^e)\times [E_n, E_{n+1})$ in the functional graph of baker's map~(\ref{eq:Gbaker}),
the cardinality of set
\begin{equation}
\label{set:X1}
\left\{ X_1 \Big\vert R\left(\frac{1}{p_{n}} (X_{1}-E_{n})\right)=Y_1, X_1\in [E_{n}, E_{n+1}) \right\}
\end{equation}
is equal to zero or one. If the cardinality is equal to one, $\Y$ is a non-leaf node.
\end{Property}
\begin{proof}
Assume the cardinality of set~(\ref{set:X1}) is
greater than one. 
There are two different elements $X_{1, i}$ and $X_{1, j}$ in set~(\ref{set:X1}) satisfying $X_{1, i} > X_{1, j}$. Thus 
$R\left(\frac{1}{p_{n}} (X_{1}-E_{n})\right)=R\left(\frac{1}{p_{n}} (X_{1}-E_{n})\right)$, that is, 
$\frac{1}{p_{n}} (X_{1, i} -E_n)-\frac{1}{p_{n}} (X_{1, j} -E_{n})=\frac{X_{1, i} -X_{1, j}}{p_{n}} < 1$.
Because $p_{n}<1$, one can obtain $X_{1, i} -X_{1, j} < 1$, which contradicts to that $X_{1, i}, X_{1, j}\in \mathbb{Z}$ are two different elements. Thus, all elements in set~(\ref{set:X1}) are the same and its cardinality equals zero or one. If the cardinality is equal to one, there is $X_1$ satisfying 
\begin{equation}\label{eq:RY1}
    R\left(\frac{1}{p_{n}} (X_{1}-E_{n})\right)=Y_1. 
\end{equation}
Next, we analyze $X_2$.
Since $\frac{Y_2+1-E_n}{p_n}-\frac{Y_2-E_n}{p_n}>1$ and considering the range of $X_2$, there is $X_2$ satisfying $X_2 \in [\frac{Y_2-E_n}{p_n}, \frac{Y_2+1-E_n}{p_n})$, namely $R\left(p_{n}\cdot X_2+ E_n\right)=Y_2$.
Combining Eqs.~\eqref{eq:GBakerp} and~\eqref{eq:RY1}, there is $\X$ satisfying $\Y = \mathrm{B}_e(\X)$, namely $\Y$ is non-leaf node.
\end{proof}

\begin{Property}
\label{prop:X_2}
When the cardinality of set~(\ref{eq:SYs}), $|\s_n|$, corresponding to a non-leaf node $\Y=(Y_1, Y_2)$, it satisfies
\begin{equation}
\label{eq:num:SYs}
|\s_n|\in
\left\{
\mathrm{R}\left(\frac{|I_n|}{p_n}\right),
\mathrm{R}\left(\frac{|I_n|}{p_n}\right)+1
\right\},
\end{equation}
where $|I_n|$ is the cardinality of $I_n=[Y_2, Y_2+1)\bigcap [E_n, E_{n+1})$.
\end{Property}
\begin{proof}
First, as for any $\X\in \s_n$, there is only one possible value of $X_1$ corresponding to
$\Y$ from Property~\ref{le:X1Y1i}. 
Then, according to Eq.~\eqref{eq:Y'} and $\mathrm{B}(x, y)$ is
bijective, the sufficient and necessary condition for $\X\in [E_n, E_{n+1})\times [0, 2^e)$ is $\Y' \in [0, 2^e)\times [E_n, E_{n+1})$.
So, from the above argument and the definition of $\s_n$, one has $|\s_n|=|\{X_2 \mid \mathrm{R}(p_n\cdot X_2+E_n) =Y_2, p_n\cdot X_2+E_n\in [E_n, E_{n+1})\}|$.
The equation in the previous set means $p_n\cdot X_2+E_n\in [Y_2, Y_2+1)$. Hence,
$p_n\cdot X_2+E_n\in I_n$.
Thus, one can calculate 
\begin{equation}\label{eq:SYs=X}
 |\s_n|=\bigg|\bigg\{\X\ \Big\vert 
 X_2\in \left[\frac{I^{\inf}_n-E_n}{p_n}, \frac{I^{\sup}_n-E_n}{p_n}\right),
 X_2\in\mathbf{Z} \bigg\}\bigg|, 
\end{equation}
where $I^{\inf}_n$ and $I^{\sup}_n$
are the infimum and supremum of interval $I_n$, respectively.
Let
$\s_n=\{\X_i\}^{|\s_n|}_{i=1}=\{(X_{1, i}, X_{2, i})\}^{|\s_n|}_{i=1}$
and
\begin{equation}
\label{eq:xminmax}
 \left\{
 \begin{split}
 X_{2, \inf} &= \inf \left(\{X_{2, i}\}_{i=1}^{|\s_n|} \right), \\
 X_{2, \sup} &= \sup \left(\{X_{2, i}\}_{i=1}^{|\s_n|} \right).
 \end{split}
 \right.
\end{equation}
One can further deduce that
\begin{equation}
\label{eq:X2min}
X_{2, \inf}\in \left[\frac{I^{\inf}_n-E_n}{p_n}, \frac{I^{\inf}_n-E_n}{p_n}+1\right)
\end{equation}
and
\begin{equation}
\label{eq:X2max}
X_{2, \sup}\in \left[\frac{I^{\sup}_n-E_n}{p_n}-1, \frac{I^{\sup}_n-E_n}{p_n}\right).
\end{equation}
From Eq.~\eqref{eq:SYs=X}, it follows that 
\begin{equation}\label{eq:SnXsupinf}
    |\s_n|=X_{2, \sup}-X_{2, \inf}+1, 
\end{equation}
then
\begin{equation*}
|\s_n| \in\left(\frac{|I_n|}{p_n}-1, \frac{|I_n|}{p_n}+1\right).
\end{equation*}
Considering the properties of the adopted quantization function $\mathrm{R}(\cdot)$ one can obtain 
\begin{equation*}
|\s_n| \in\left(R\left(\frac{|I_n|}{p_n}\right)-1, R\left(\frac{|I_n|}{p_n}\right)+2\right).
\end{equation*}
It follows from $|\s_n|$ is an integer that relation~\eqref{eq:num:SYs} holds.
\end{proof}

\begin{Corollary}\label{coro:d1}
Given a non-leaf node $\Y$ satisfying $[Y_2, Y_2+1)\in[E_{n}, E_{{n}+1})$, its in-degree in the functional graph of baker's map~(\ref{eq:Gbaker}), $d$, satisfies
\begin{equation*}
d\in\left\{\mathrm{R}\left(\frac{1}{p_{n}}\right) ,\mathrm{R}\left(\frac{1}{p_{n}}\right)+1\right\}.
\end{equation*}
\end{Corollary}
\begin{proof}
    When $[Y_2, Y_2+1)\subset [E_{n}, E_{n+1})$, one has $\Y' \in [0, 2^e)\times [E_{n}, E_{{n}+1})$. 
It means $\X\in [E_{n}, E_{{n}+1})\times [0, 2^e)$. So, $\s\subset\s_{n}$. 
From the definition of sets $\s$ and $\s_{n}$, one has 
$\s_{n}\subset\s$. Thus $\s_{n}=\s$.
According to Property~\ref{prop:X_2}, in-degree $d$ satisfies
$$d=|\s|=|\s_{n}| \in \left\{\mathrm{R}\left(\frac{1}{p_{n_1}}\right), \mathrm{R}\left(\frac{1}{p_{n_1}}\right)+1\right\}.$$
\end{proof}

\begin{Corollary}
\label{coro:|Sn|}
When the cardinality of set~(\ref{eq:SYs}), $|\s_n|\neq 0$, corresponding to a non-leaf node $\Y=(Y_1, Y_2)$, it satisfies
\begin{equation}
\label{eq:num:SYsIn}
|\s_n|=
\begin{cases}
 \left\lceil \frac{|I_n|}{p_n} \right\rceil & \mbox{if } I_n=[E_{n}, Y_2+1);\\
 \mathrm{R}\left(\frac{|I_n|}{p_n}\right) & \mbox{if } I_n= [Y_2, E_{n+1});\\
 2^e & \mbox{if } I_n=[E_{n}, E_{n+1}),
\end{cases}
\end{equation}
where $\lceil\cdot \rceil$ is the ceiling function.
\end{Corollary}
\begin{proof}
According to Eqs.~\eqref{eq:X2min} and~\eqref{eq:X2max}, one has
\begin{equation}
\label{eq:X_2min2}
 X_{2, \inf}=
 \begin{cases}
 0 & \mbox{if } I_n\in \{[E_{n}, Y_2+1), [E_{n}, E_{n+1})\};\\
 \left\lceil\frac{I_n^{\inf}-E_n}{p_n} \right\rceil & \mbox{if } I_n= [Y_2, E_{n+1}),
 \end{cases}
\end{equation}
and
\begin{equation}
\label{eq:X_2max2}
X_{2, \sup}=
\begin{cases}
 \left\lceil\frac{|I_n|}{p_n} -1 \right\rceil & \mbox{if } I_n=[E_{n}, Y_2+1);\\
 2^e - 1 & \mbox{if } I_n\in\{ [Y_2, E_{n+1}), [E_{n}, E_{n+1})\}.
 \end{cases}
\end{equation}
If $I_n= [Y_2, E_{n+1})$, one can get
\[
\begin{split}
 \left\lceil\frac{I_n^{\inf} - E_n}{p_n} \right\rceil
 &= \left\lceil\frac{E_{n+1}-E_n+I_n^{\inf}-E_{n+1}}{p_n}\right\rceil\\
 &= 2^e+\left\lceil\frac{-|I_n|}{p_n} \right\rceil\\
 &= 2^e - \mathrm{R}\left(\frac{|I_n|}{p_n}\right).
\end{split}
\]
Then, it follows from Eqs.~\eqref{eq:X_2min2} and~\eqref{eq:X_2max2} that
\begin{multline*}
|X_{2, \sup} -X_{2, \inf}|=\\
 \begin{cases} \left\lceil\frac{|I_n|}{p_n} -1 \right\rceil & \mbox{if } I_n=[E_{n}, Y_2+1);\\
 \mathrm{R}\left(\frac{|I_n|}{p_n}\right)- 1 & \mbox{if } I_n= [Y_2, E_{n+1});\\
 2^e - 1, & \mbox{if } I_n= [E_{n}, E_{n+1}).
 \end{cases}
\end{multline*}
So, Eq.~\eqref{eq:num:SYsIn} holds from Eq.~\eqref{eq:SnXsupinf}.
\end{proof}

Given a node $\Y$, the cardinality of its corresponding $\s_n$ can be estimated by Property~\ref{prop:X_2} and Corollary~\ref{coro:|Sn|}, and the number of its preimages is $\sum_{n=1}^k |\s_n|$.
According to the relation between $[Y_2, Y_2+1)$ and $\{[E_n, E_{n+1})\}_{n=1}^{k}$, there are two possible cases of $\s_n$ corresponding to $\Y$, namely $n_1=n_2$ or $n_1<n_2$, 
where $k=2^e-1$, $Y_2\in [E_{n_1}, E_{n_1+1})$, $Y_2+1\in [E_{n_2}, E_{n_2+1})$. Figure~\ref{Fig:x3kindx} illustrates the two cases
in nodes marked ``\textit{A}" and ``\textit{B}", respectively. 
Analyzing such two cases separately, one can deduce the upper bound of the in-degree of any non-leaf node in the functional graph of baker's map~(\ref{eq:Gbaker}), as shown in Proposition~\ref{prop:2n*d}.

\begin{figure}[!htb]
\centering
\begin{minipage}{\BigOneImW}
 \includegraphics[width=\BigOneImW]{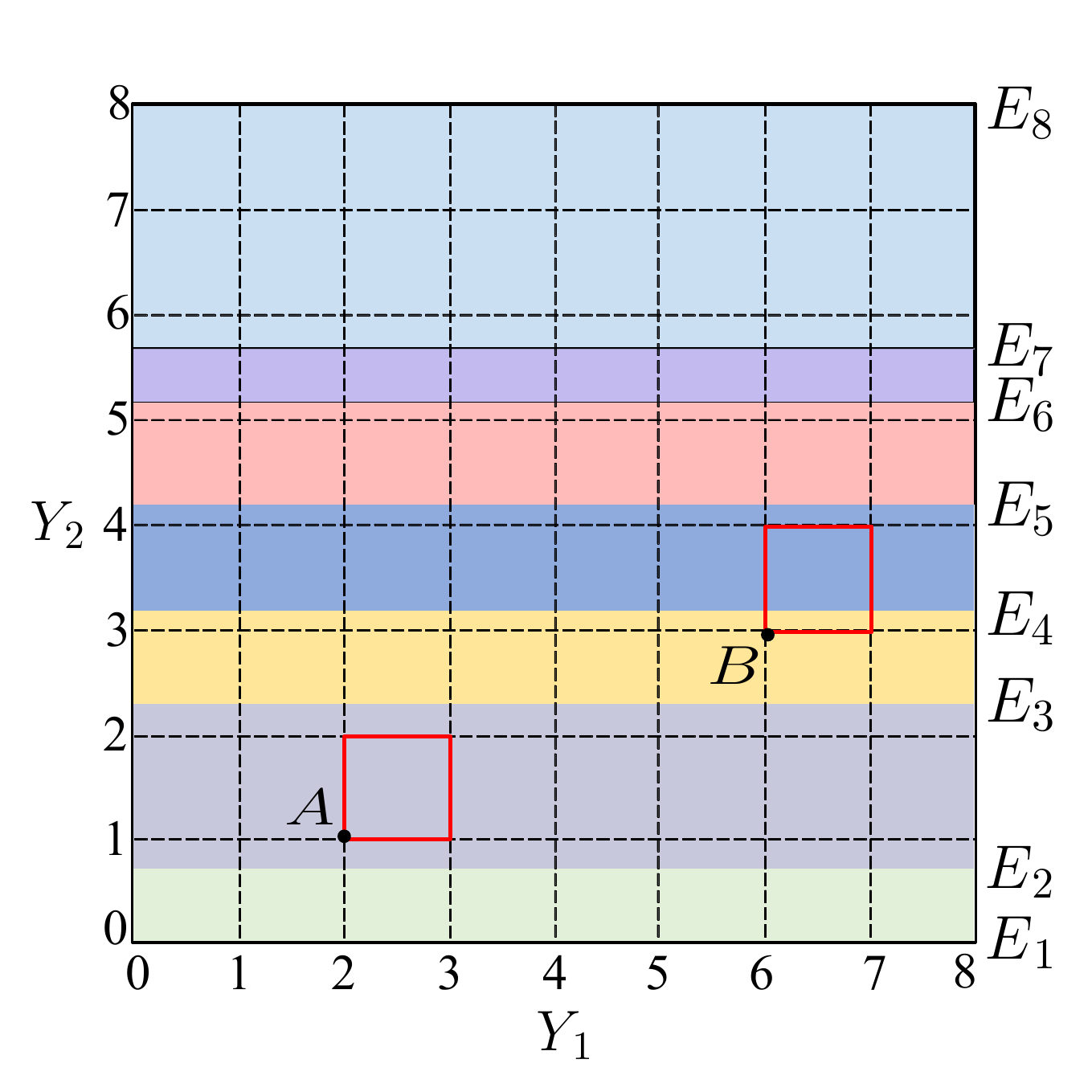}
\end{minipage}
\caption{Two possible location cases of $\Y$ with respect to $\{E_n\}_{n=1}^{k+1}$ when $e=3$.}
\label{Fig:x3kindx}
\end{figure}

\begin{figure*}[!htb]
\centering
\begin{minipage}{\ThreeImW}\hspace{1em}
\centering
\includegraphics[width=\ThreeImW]{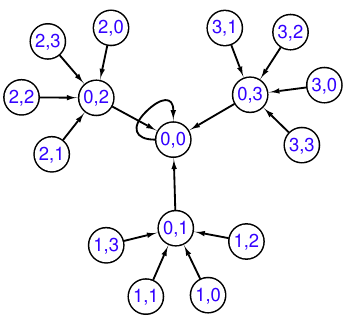}
a)
\end{minipage}
\begin{minipage}{\ThreeImW}\hspace{1em}
\centering \includegraphics[width=\ThreeImW]{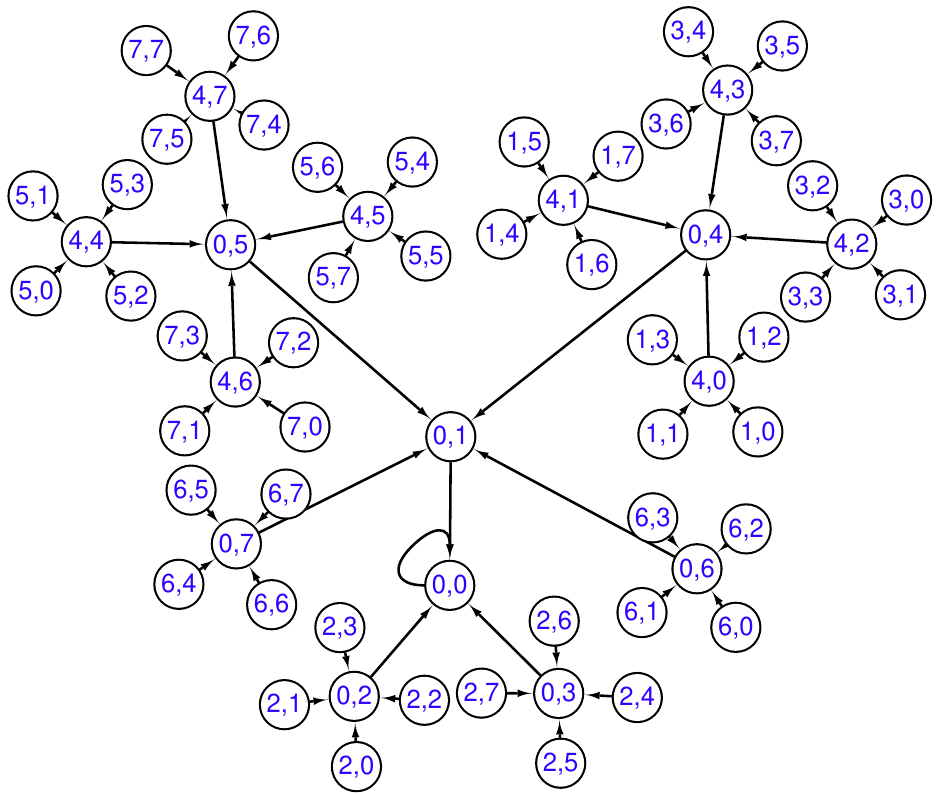}
\vspace{-1em}
b)
\end{minipage}
\begin{minipage}{\ThreeImW}
\centering \includegraphics[width=\ThreeImW]{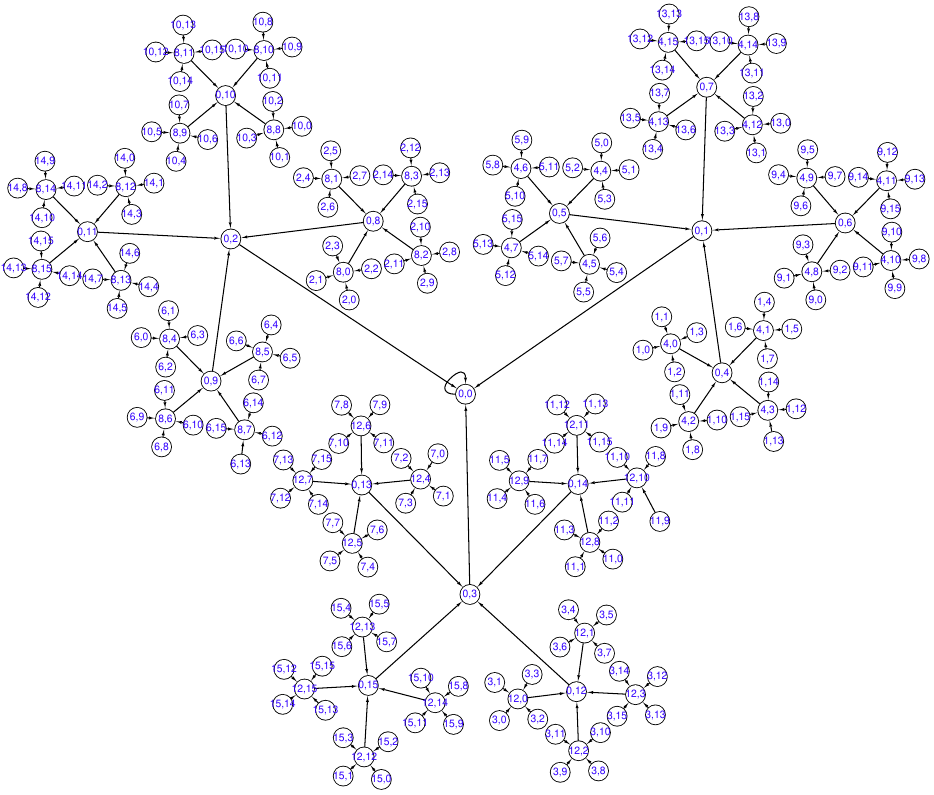}
\vspace{-0.5em}
c)
\end{minipage}
\caption{Functional graphs of the generalized baker's map with $p_n\equiv\frac{1}{4}$: a) $e=2$; b) $e=3$; c) $e=4$.}
\label{fig:perioddistributionof2222}
\end{figure*}

\begin{Proposition}
\label{prop:2n*d}
The in-degree of any non-leaf node in the functional graph of baker's map~(\ref{eq:Gbaker}), $d$, satisfies
\begin{equation}
\label{eq:prop1}
d \leq \mathrm{R}\left(\frac{1}{\min(p_1, p_2, \cdots, p_k)}\right)+1.
\end{equation}
\end{Proposition}
\begin{proof}
As for any non-leaf node $\Y=(Y_1, Y_2)$,
there exist $n_1, n_2$ such that $Y_2 \in [E_{n_1}, E_{n_1+1}), Y_2+1\in [E_{n_2}, E_{n_2+1})$.
According to the relation between $n_1$ and $n_2$, the proof is divided into the following two cases:
\begin{itemize}
\item $n_1=n_2$: this implying $[Y_2, Y_2+1)\in[E_{n_1}, E_{{n_1}+1})$. Referring to Corollary~\ref{coro:d1}, Eq.~\eqref{eq:prop1} holds.

\item $n_1<n_2$: one has $I_{n_1}=[Y_2, E_{n_1+1})$, $I_t=[E_t, E_{t+1})$ and $I_{n_2}=[E_{n_2}, Y_2+1)$, where $t\in(n_1, n_2)$.
It yields from Corollary~\ref{coro:|Sn|} that
\begin{equation*}
 |\s_{i}|=
 \begin{cases}
  \left\lceil\frac{|I_{n_2}|}{p_{n_2}} \right\rceil & \mbox{if } i=n_2; \\
 \mathrm{R}\left(\frac{|I_{n_1}|}{p_{n_1}} \right) & \mbox{if } i=n_1; \\
 2^e & \mbox{if } i \in(n_1, n_2)
 \end{cases}
\end{equation*}
when $|\s_{i}|\neq 0$
So, one has
\begin{equation*}
\begin{split}
 d &= \sum\limits_{i=n_1}^{n_2}|\s_i|\\
 &\leq \mathrm{R}\left(\frac{|I_{n_1}|}{p_{n_1}}\right)+\frac{|I_{n_1+1}|}{p_{n_1+1}}+\cdots+\left\lceil\frac{|I_{n_2}|}{p_{n_2}} \right\rceil \\
 &< \sum_{i=n_1}^{n_2}\frac{|I_i|}{p_{i}} +1.
\end{split}
\end{equation*}
\end{itemize}
As $k\geq n_2$, in both cases,
\begin{equation*}
\begin{split}
d< & \mathrm{R}\left(\frac{1}{\min(p_{n_1}, p_{n_1+1}, \cdots, p_{n_2})} \right)+1\\
 < & \mathrm{R}\left(\frac{1}{\min(p_1, p_2, \cdots, p_k)}\right)+1
 \end{split}
\end{equation*}
holds, which finishes the proof of this proposition.
\end{proof}

As depicted in Fig.~\ref{fig:perioddistributionof212} with $e=2, 3, 4$, it is evident that the largest in-degree is bounded by 
$$\mathrm{R}\left(\frac{1}{\min(\frac{2}{5}, \frac{1}{5}, \frac{2}{5})}\right)+1=6.$$
Upon comparing Fig.~\ref{fig:perioddistributionof212}a), b), and c), it is apparent that the functional graph's largest in-degree, when $e = 4$, surpasses that corresponding to $e\in {2, 3}$.
Observing the trend with increasing $e$, it is observed that the largest in-degree tends to stabilize at a certain value, accompanied by the emergence of self-similarity patterns, as depicted in Fig.~\ref{fig:perioddistributionof212}. This self-similarity is further elucidated in the in-degree distribution of the functional graph, as expounded in Proposition\ref{prop:in-degreedis1} and Corollary~\ref{coro:in-degreedis2}.

\begin{Proposition}\label{prop:in-degreedis1}
As $e$ increases, the ratio of the number of nodes in $[0,2^e)\times [E_n, E_{n+1})$ with in-degree $\mathrm{R}(\frac{1}{p_n})$ and $\mathrm{R}(\frac{1}{p_n}) + 1$ to the total number of nodes in the functional graph of baker's map~(\ref{eq:Gbaker}) approach $\left(\mathrm{R}\left(\frac{1}{p_n}\right) + 1\right)(p_n^2)-p_n$ and $p_n-p_n^2\mathrm{R}(\frac{1}{p_n})$, respectively.
\end{Proposition}
\begin{proof}
For any node $\Y \in [0,2^e)\times [E_n, E_{n+1})$, if $\Y \in [0, 2^e)\times [\mathrm{R}(E_n) + 1, \mathrm{R}(E_{n+1})-1]$, this implies $\Y' \in [0, 2^e)\times [E_n, E_{n+1})$, which is the necessary and sufficient condition for $\X \in [E_n, E_{n+1})\times [0, 2^e)$. Let $N_{e}'$ denote the number of nodes with a fixed first coordinate and a second coordinate in $[\mathrm{R}(E_n) + 1, \mathrm{R}(E_{n+1})-1]$ in $\mathbb{F}_{e}$. Then,
$N_{e}'=\mathrm{R}(2^{e}F_{n+1})-\mathrm{R}(2^{e}F_n)-2$
follows from $E_n=2^eF_n$. 
According to Property~\ref{le:X1Y1i}, $\Y$ is a non-leaf node if there exists an injective $Y_1=\mathrm{R}\left(\frac{1}{p_n}(X_1-E_n)\right)$ for any $X_1 \in [E_n, E_{n+1})$. Consequently, there are $\mathrm{R}(E_{n+1})-\mathrm{R}(E_n)$ possible values for $Y_1$, resulting in $N_{e}'(N_{e}' + 2)$ non-leaf nodes.
As $e$ increases, $N_{e}'$ approaches $2^e p_n$, leading to approximately $(2^e p_n)^2$ non-leaf nodes in $[0,2^e)\times [E_n, E_{n+1})$. Similarly, there are about $2^{2e} p_n(1-p_n)$ leaf nodes in this interval. For nodes outside $[0,2^e)\times [\mathrm{R}(E_n) + 1, \mathrm{R}(E_{n+1})-1]$, their in-degree is less straightforward, but the proportion of such nodes to the total number becomes negligible as $e$ increases.
Referring to Corollary~\ref{coro:d1}, the in-degree of all non-leaf nodes in $[0,2^e)\times [E_n, E_{n+1})$ is either $\mathrm{R}(\frac{1}{p_n})$ or $\mathrm{R}(\frac{1}{p_n}) + 1$. Denoting $N_1$ as the number of nodes with in-degree $\mathrm{R}(\frac{1}{p_n})$, the number of nodes with in-degree $\mathrm{R}(\frac{1}{p_n}) + 1$ approximates $2^{2e}p_n^2-N_1$. With the necessary and sufficient condition for $\Y' \in [0, 2^e)\times [E_n, E_{n+1})$ being $\X \in [E_n, E_{n+1})\times [0, 2^e)$, approximately $2^{2e}p_n$ nodes are the preimages of non-leaf nodes in $[0,2^e)\times [E_n, E_{n+1})$. Therefore,
\[
N_1 \approx \left(\mathrm{R}\left(\frac{1}{p_n}\right) + 1\right)(2^{2e}p_n^2)-2^{2e}p_n.
\]
The ratio of the number of nodes with in-degree $\mathrm{R}(\frac{1}{p_n})$ in $[0,2^e)\times [E_n, E_{n+1})$ to the total number of nodes in $\mathbb{F}_e$ approaches $p_n^2\mathrm{R}(\frac{1}{p_n})-(1-p_n)$. Similarly, the ratio for nodes with in-degree $\mathrm{R}(\frac{1}{p_n}) + 1$ 
approaches $p_n-p_n^2\mathrm{R}(\frac{1}{p_n})$.
\end{proof}

\begin{Corollary}\label{coro:in-degreedis2}
The distribution of any given in-degree in the functional graph of baker's map~(\ref{eq:Gbaker}) approaches a constant as $e$ increases.
\end{Corollary}

From Proposition~\ref{prop:in-degreedis1}, for any $n$, the nodes in $[0,2^e)\times [E_n, E_{n+1})$ follow similarly regularity as $e$ increase. Then, under under the specific parameter where all parameters $p_n$ are equal and equal to a power of two, all the nodes exhibit the same regularity as $e$ increases.
In this case, the in-degree for any non-leaf node of the corresponding functional graph is a constant, as described in Property~\ref{prop:degree}, namely the functional graph of baker's map~(\ref{eq:Gbaker}) demonstrates the obvious semi-fractal property.
When parameters $p_n\equiv \frac{1}{2^{2}}$, the functional graph of baker's map~(\ref{eq:Gbaker}) with a precision larger than two 
can be considered as the evolved version from $\mathbb{F}_2$ (shown in Fig.~\ref{fig:perioddistributionof2222}a) in a  semi-fractal way. 
The corresponding $\mathbb{F}_3$ and $\mathbb{F}_4$ are depicted in 
Fig.~\ref{fig:perioddistributionof2222}b) and c), respectively. 

\begin{Property}
If $p_n=\frac{1}{2^{e^*}}$ for any $n$, the in-degree of any non-leaf node in the functional graph of baker's map~(\ref{eq:Gbaker}) with precision $e\geq e^*$ is $2^{e^*}$, 
where $e^*$ is a positive integer.
\label{prop:degree}
\end{Property}
\begin{proof}
When $p_n=\frac{1}{2^{e^*}}$ for any $n \in \{1, 2, \cdots, k\}$, 
\begin{equation}\label{eq:2D:En}
 E_n=2^e\cdot \sum_{i=1}^{n-1} p_i=(n-1)2^{e-e^*}
\end{equation}
is an integer in the fixed-point arithmetic domain with precision $e\geq e^*$.
Assume $Y_2\in[E_n, E_{n+1})$, one has $Y_2\in [E_n, E_{n+1}-1]$. It follows from $Y_2$ is an integer that 
\begin{equation}
\label{eq:Y2Y2+1}
 [Y_2, Y_2+1)\subseteq [E_n, E_{n+1}).
\end{equation}
In such case, Eq.~\eqref{eq:SYs=X} becomes
\begin{multline*}
\s_n=\bigg\{\X \mid X_2 \in \\
 [2^{e^*}Y_2 - 2^{e^*}E_n, 2^{e^*}(Y_2+1) - 2^{e^*}E_n),
X_2\in\mathbf{Z} \bigg\}.
\end{multline*}
Referring to Eq.~\eqref{eq:xminmax} and the above relation, one obtains $X_{2, \min}= 2^{e^*}Y_2-2^{e^*}E_n$ and $X_{2, \sup}= 2^{e^*}(Y_2+1) -2^{e^*}E_n -1$.
Thus it follows that $|\s_n|=X_{2, \sup}-X_{2, \min}+1=2^{e^*}$.
Then, it follows from Eq.~\eqref{eq:Y2Y2+1} that $\Y' \in [0, 2^e)\times [E_n, E_{n+1})$. Because the sufficient and necessary condition for $\Y' \in [0, 2^e)\times [E_n, E_{n+1})$ is $\X\in [E_n, E_{n+1})\times [0, 2^e)$, one can obtains
$d=|\s|=|\s_n|=2^{e^*}$ from the definition of $|\s|$ and $|\s_n|$.
\end{proof}

The self-similarity inherent in the functional graph of baker's map is quantified by the fractal dimension, referred to as the self-similar exponent, considered at different scales corresponding to the number of iterations. baker's map is characterized by a piecewise function, which partitions the domain into $k^i$ intervals. Notably, each orbit terminating within the node in a specific interval exhibits identical iteration expressions after the $i$-th iteration. Thus, the orbits of length $i$ in the functional graph can be categorized into $k^i$ sets, where $k$ signifies the number of intervals in Eq.~(\ref{eq:Baker}).
Proposition~\ref{prop:in-degreedis1} elucidates that the average in-degree of all non-leaf nodes in the functional graph of baker's map (Eq.~(\ref{eq:Gbaker})) remains constant $d_a$. Consequently, the cardinality of orbits tends to $d_a^i$ in the functional graph post $i$ iterations. The fractal dimension, denoted as $D_B$, is expressed as
\[
D_B= -\frac{\log(k)}{\log(d_a)}.
\]
Furthermore, the structural composition of the functional graph of baker's map entails both cycles and trees, with the cycles constituting a fractional portion of the overall network. The tree structure exhibits characteristics akin to a fractal tree. The average length of the entry cycle can be approximated by considering the aggregate length of all cycles in the network and the size of the entire network. When $p_n\equiv \frac{1}{k}$, one has $D_B = 1$ and the tree adheres to a strictly fractal pattern. 
For the scenario where $p_n\equiv \frac{1}{2^{e^*}}$, Property~\ref{prop:degree} demonstrates that the in-degree of any non-leaf node in the functional graph of baker's map remains a constant. 
Properties~\ref{prop:e-e+12X_1} and~\ref{prop:e-e+12X_1+1} unveil that each tree, comprising $1+2^{e^*}$ nodes
in $\mathbb{F}_{e}$, corresponds to four isomorphic trees in $\mathbb{F}_{e+1}$. Notably, the root nodes of two of these trees exhibit a systematic displacement pattern.
The evolution of the functional graph $\mathbb{F}_{e+1}$ from $\mathbb{F}_{e}$ is illustrated in Figure~\ref{fig:BakerMapXYE-E+1}, depicting the case for $e^*=2$.
This phenomenon is identified as the emergence mechanism of semi-fractals in $\mathbb{F}_{e}$ with increasing implementation precision $e$. an analysis further expounded upon in Sec.~\ref{sec:HDBM}.

\begin{Property}
\label{prop:e-e+12X_1}
If $p_n\equiv \frac{1}{2^{e^*}}$ and $\mathrm{B}_e(X_1, X_2)=\Y$, one has
\begin{multline}
\mathrm{B}_{e+1}(2X_1, 2X_2)= \mathrm{B}_{e+1}(2X_1, 2X_2+1)=\\*%
 \begin{cases}
 (2Y_1, 2Y_2)   & \mbox{if } \Delta\in[0, \frac{1}{2p_n});\\
 (2Y_1, 2Y_2+1) & \mbox{if } \Delta\in[\frac{1}{2p_n}, \frac{1}{p_n}),
 \end{cases}
\label{eq:2D:Be+12X1} 
\end{multline}
where $e \geq e^*$ and $\Delta=X_2 + 2^{e}(n-1)-\frac{Y_2}{p_n}$.
\end{Property}
\begin{proof}
Referring to Eq.~\eqref{eq:GBakerp}, one can get 
\begin{multline}\label{eq:2x12x2e+1}
\mathrm{B}_{e+1}(2X_1, 2X_2)=\\
 \mathrm{R}\left(2^{e^*+1}X_1-2^{e+1}(n-1),\frac{2X_2+2^{e+1}(n-1)}{2^{e^{*}}} \right),
\end{multline}
and 
\begin{multline*}
\mathrm{B}_{e+1}(2X_1, 2X_2+1)=\\
 \mathrm{R}\left(2^{e^*+1}X_1-2^{e+1}(n-1),\frac{2X_2+1+2^{e+1}(n-1)}{2^{e^{*}}} \right).
\end{multline*}
It follows from $e \geq e^*$ and $\mathrm{R}\left(\frac{2X_2}{2^{e^{*}}}\right) = R\left(\frac{2X_2+1}{2^{e^{*}}} \right)$ that $\mathrm{R}\left(\frac{2X_2+2^{e+1}(n-1)}{2^{e^{*}}}\right) = \mathrm{R}\left(\frac{2X_2+1+2^{e+1}(n-1)}{2^{e^{*}}} \right)$, that is, 
\begin{equation}\label{eq:2x2+1=2X_2}
\mathrm{B}_{e+1}(2X_1, 2X_2) = \mathrm{B}_{e+1}(2X_1, 2X_2+1).
\end{equation}
Then, according to Eq.~\eqref{eq:GBakerp}, one has
\begin{equation}
\label{eq:Y_1Y_2X1X2}
     \begin{cases}
     Y_1 = 2^{e^*}X_1 -(n-1)2^{e}\\
     Y_2 = \mathrm{R}\left( X_2/2^{e^*} \right)+ (n-1)2^{e-e^*}.
     \end{cases}
\end{equation}
Referring to Eq.~\eqref{eq:2x12x2e+1}, one can get
\begin{equation*}
\mathrm{B}_{e+1}(2X_1, 2X_2)=\left(2Y_1, \mathrm{R}(2Y_2 + 2p_n\Delta)\right).
\end{equation*}
Because $p_n\Delta=X_2/2^{e^*}+ (n-1)2^{e-e^*} -Y_2 \in [0,1)$, one can get 
\begin{equation*}
\mathrm{R}(2Y_2 + 2p_n\Delta) =
 \begin{cases}
 2Y_2   & \mbox{if } \Delta\in[0, \frac{1}{2p_n});\\
 2Y_2+1 & \mbox{if } \Delta\in[\frac{1}{2p_n}, \frac{1}{p_n}).
 \end{cases}
\end{equation*}
So, one can get Eq.~\eqref{eq:2D:Be+12X1} from the above equation and Eq.~\eqref{eq:2x2+1=2X_2}.
\end{proof}

\begin{Property}
\label{prop:e-e+12X_1+1}
If $p_n\equiv \frac{1}{2^{e^*}}$ and $\mathrm{B}_e(X_1, X_2)=\Y$, one has
\begin{multline*}
\mathrm{B}_{e+1}(2X_1+1, 2X_2)= \mathrm{B}_{e+1}(2X_1+1, 2X_2+1)=\\
 \begin{cases}
 (2Y_1+2^{e^*+1}, 2Y_2) & \mbox{if } \Delta\in[0, \frac{1}{2p_n});\\
 (2Y_1+2^{e^*+1}, 2Y_2+1) & \mbox{if } \Delta \in[\frac{1}{2p_n}, \frac{1}{p_n}),
 \end{cases} 
\end{multline*}
where $\Delta=X_2 + 2^{e}(n-1)-\frac{Y_2}{p_n}$ and $e\geq e^*$.
\end{Property}
\begin{proof}
The proof is very similar to that of Property~\ref{prop:e-e+12X_1} and omitted here.
\end{proof}

\begin{figure}[!htb]
\centering
\begin{minipage}{1.1\BigOneImW}
 \includegraphics[width=1.1\BigOneImW]{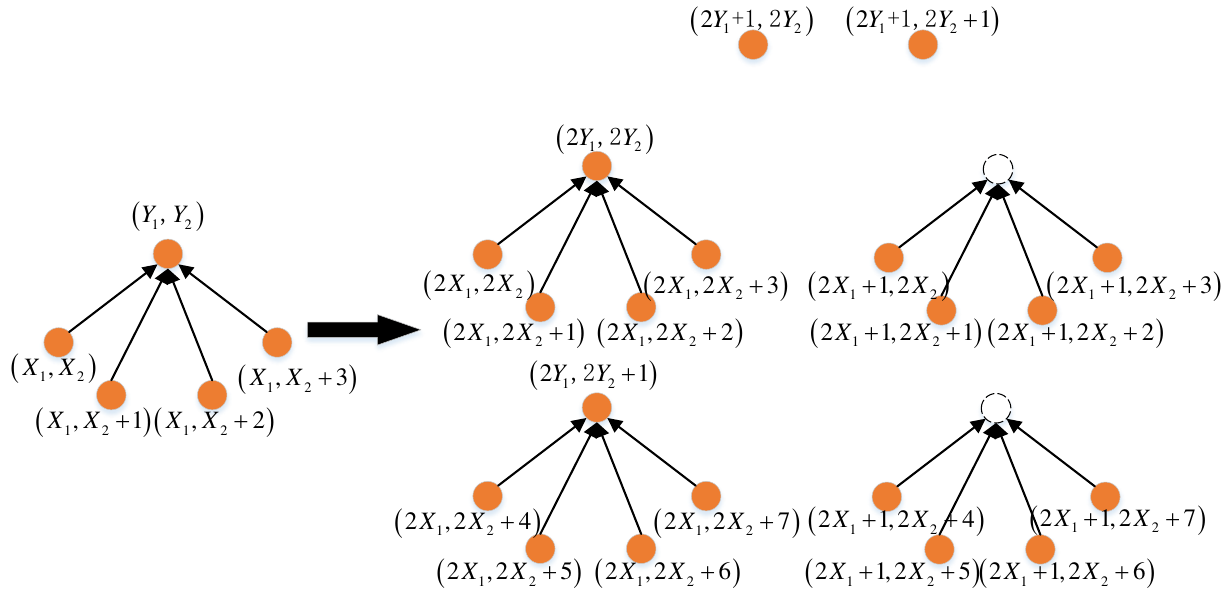}
\end{minipage}
\caption{The mapping relation between functional graph of baker's map~(\ref{eq:Gbaker}) with precision $e$ and that with precision $e+1$.}
\label{fig:BakerMapXYE-E+1}
\end{figure}

\section{The functional graph of the higher-dimensional baker's map in digital domain}
\label{sec:HDBM}

In \cite{OZTURK2018395}, \"{O}zt\"{u}rk and Kili\c{c} 
replaced the least significant zero bits of every state of some
special chaotic maps,
whose implementation only involves bitwise logic operators, 
by random ones to produce their true chaotic orbits in a digital domain.
In \cite{Ozturk:Baker:TCASI2019}, they further extend the idea to
every dimension of a higher-dimensional baker's map.
It was proved to satisfy Devaney's definition of chaos as \cite{wangqx:higher:TCASI21}. 
This section analyzes the graph structure of this map implemented in a digital domain with fixed-point and floating-point number representation.

Given a dimension $m$, the corresponding HDBM, $\mathcal{B}: [0, 1)^m$ $\rightarrow [0, 1)^m$, can be represented as
\begin{equation}
\label{eq:BH}
\mathcal{B}(\x)=
 \begin{cases}
 \mathrm{T}_0(\x) & \mbox{if }\x\in C_0; \\
 \mathrm{T}_1(\x) & \mbox{if }\x\in C_1; \\
\hspace{1em}\vdots & \hspace{2em}\vdots \\
 \mathrm{T}_{2^{m-1}-1}(\x) & \mbox{if }\x\in C_{2^{m-1}-1},
 \end{cases}
\end{equation}
where $\x=(x_1, x_2, \cdots, x_m)^\intercal$,
\begin{multline*}
C_k=\left\{\x\ \Big\vert\ x_i \in \left[\frac{\beta(k, i)}{2}, \frac{\beta(k, i)+1}{2}\right), i=1\sim m-1; \right.\\
x_m \in [0, 1) \bigg\}, 
\end{multline*}
affine transformation
\begin{equation}
\label{eq:Tkx}
\mathrm{T}_k(\x)= {\bf A}\cdot \x-{\bf p}_k,
\end{equation}
\begin{equation*}
{\bf A}=\begin{bmatrix}
 2 & 0 & \cdots & 0 & 0 \\
 0 & 2 & \cdots & 0 & 0 \\
 \vdots & \vdots & \ddots & \vdots & \vdots \\
 0 & 0 & 0 & 2 & 0 \\
 0 & 0 & 0 & 0 & \frac{1}{2^{m-1}}
 \end{bmatrix},
\end{equation*}
\begin{equation*}
{\bf p}_{\rm k}=
 \begin{bmatrix}
 \beta(k, 1) \\
 \beta(k, 2) \\
 \vdots \\
 \beta(k, m-1) \\
 -k/2^{m-1}
 \end{bmatrix},
\end{equation*}
$\beta(k, i)$ outputs the $i$-th significant bit of $k$ and
$k\in\{0, 1$, $\cdots, 2^{m-1}-1\}$.
The schematic demonstration of the transformation of HDBM with $m=3$ is shown in Fig.~\ref{fig:HDBM}. When $\x\in C_k$, one has $\mathrm{B}_{\rm H}(\x) \in D_k$,
where
\begin{multline*}
D_k=\biggl\{{\bf y}\ \Big\vert\ y_m \in \left[\frac{k}{2^{m-1}}, \frac{k+1}{2^{m-1}}\right)\\
 y_i \in [0, 1), i=1\sim m-1 \biggr\}.
\end{multline*}
To ensure that the structure of the functional graph of HDBM is not completely degraded in the arithmetic domain, its precision $e$ should be larger than $m$.

\begin{figure}[!htb]
\centering
\begin{minipage}{1\BigOneImW}
\includegraphics[width=1\BigOneImW]{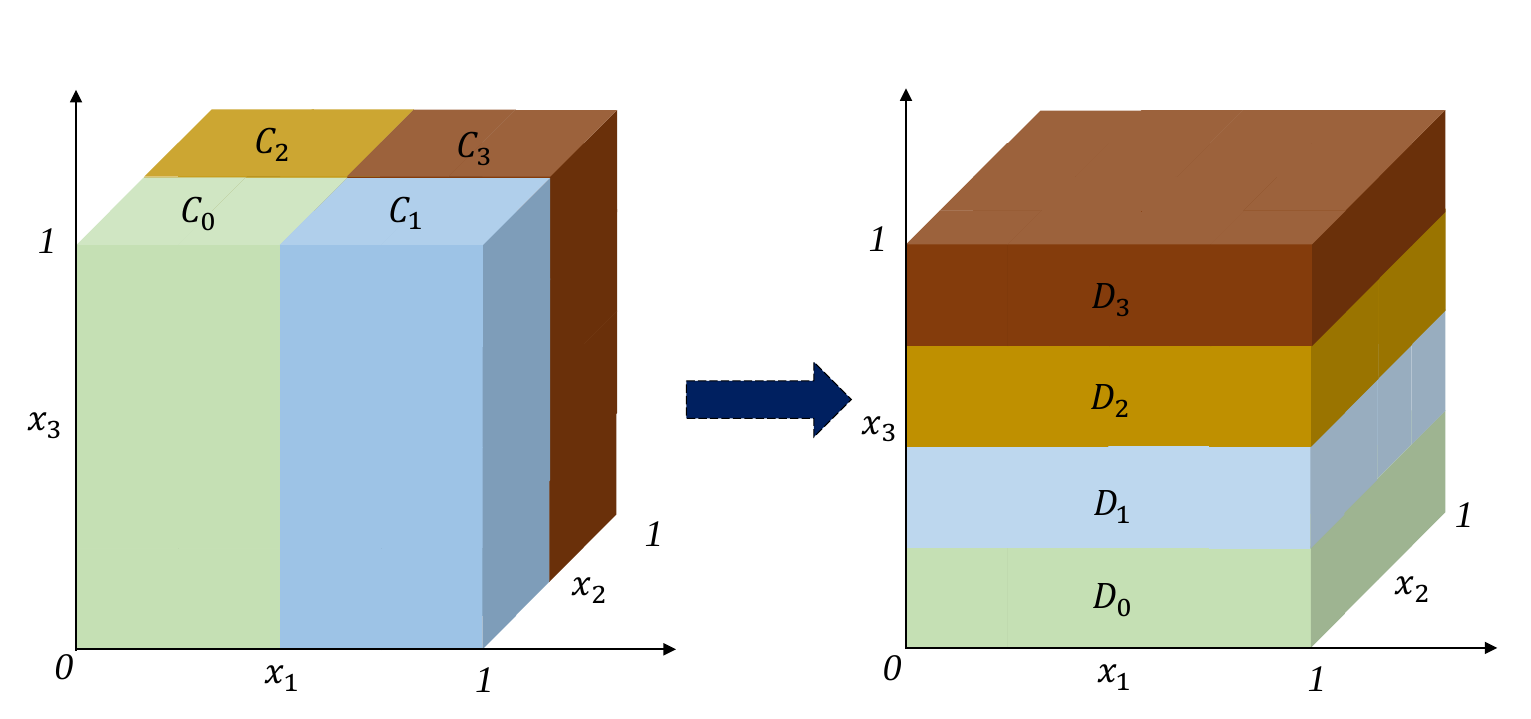}
\end{minipage}
\caption{Transformation demonstration of the higher-dimensional baker's map~(\ref{eq:BH}).}
\label{fig:HDBM}
\end{figure}

\subsection{The functional graph of HDBM in a fixed-point arithmetic domain}
\label{subsec:HDBMFix}

In the fixed-point arithmetic domain with precision $e$, the domain and range of HDBM~(\ref{eq:BH}) are both discrete set
$\{{\bf x} \mid x_1=\frac{X_1}{2^e}, x_2=\frac{X_2}{2^e}, \cdots,
x_m=\frac{X_m}{2^e}, X_i\in\mathbb{Z}_{2^e}\}$.
Hence, it can be represented as
\begin{equation}
\label{eq:B_He}
\mathcal{B}_{\rm e}(\X)=\mathrm{R}(2^e\cdot\mathcal{B}(\X/2^{e})),
\end{equation}
where $\X=(X_1, X_2, \cdots, X_m)^\intercal$.
And for $\X\in C'_k$, one has $\Y=\mathcal{B}_{\rm e}(\X) \in D'_k$, where 
$$ C'_k=\left\{\X\ \Big\vert\ \frac{\X}{2^e} \in C_k \right\}$$ 
and 
$$D'_k= \left\{\Y\ \Big\vert\ \frac{\Y}{2^e} \in D_k \right\}.$$
Let $\mathbb{F}_{\rm H, e}$ denote the functional graph corresponding to HDBM~(\ref{eq:B_He}).
As shown in Property~\ref{prop:indgreem}, the in-degree of any non-leaf node of
$\mathbb{F}_{\rm H,e}$ is a constant related with $m$ when the arithmetic precision $e$ is large enough.

\begin{Property}
\label{prop:indgreem}
When $e\geq m$, in-degree of any non-leaf node in $\mathbb{F}_{\rm H,e}$ is $2^{m-1}$.
\end{Property}
\begin{proof}
As for any non-leaf node $\Y\in D'_k$, $k$ and $\tilde{{\bf p}}_{\rm k}$ are well-determined. According to Eqs.~\eqref{eq:BH} and ~\eqref{eq:B_He}, one has $\Y = R({\bf A}\cdot \X-2^e{\bf p}_k)$. Since $\bf{A}$ is a diagonal matrix, the above equation can be express as 
\begin{equation}
 \label{eq:H,e:Ym}
 Y_m= \mathrm{R}\left(\frac{X_m}{2^{m-1}}+2^{e-m+1}k\right)
\end{equation}
and 
\begin{equation}
\label{eq:YXsub}
\begin{split}
 \tilde{\Y} &= \mathrm{R}(2 {\bf I}_{m-1}\cdot \tilde{\X}
 -2^e\tilde{{\bf p}}_{\rm k}) \\
            &=2\tilde{\X} -2^e\tilde{{\bf p}}_{\rm k},
 \end{split}
\end{equation}
where  ${\bf I}_{m-1}$ is an identity matrix of order $(m-1)\times (m-1)$, $\tilde{\Y}=(Y_1, Y_2, \cdots, Y_{m-1})^\intercal$, $\tilde{\X}=(X_1, X_2, \cdots, X_{m-1})^\intercal$ and $\tilde{{\bf p}}_{\rm k}=(\beta(k, 1)$, $\beta(k, 2), \cdots, \beta(k, m-1))^\intercal$.
Define the set containing all its preimages as $\s=\{\X\ \vert\ \mathcal{B}_{\rm e}(\X )=\Y \}$. 
For a given $\Y$, $k$ and $\tilde{{\bf p}}_{\rm k}$ are well-determined. one can see that $\tilde{\X}$ is also well-determined via Eq.~(\ref{eq:YXsub}) for any $\Y$ in a fixed-point arithmetic domain, that is, the cardinality of $\s$ is only determined by $Y_m$.
So, referring to Eq.~\eqref{eq:H,e:Ym}, one has 
\begin{equation*} 
 |\s|=\left|\left\{X_m\ \Big\vert\ \mathrm{R}\left(\frac{X_m}{2^{m-1}}+2^{e-m+1}k\right)=Y_m\right\}\right|.
\end{equation*}
As $k\in \mathbb{Z}$, the above equation can be simplified as
\begin{equation*}
 |\s|= \left|\left\{X_m\ \big\vert \ X_m\in\left[\hat{Y}_m, \hat{Y}_m+2^{m-1}\right)\right\}\right|,
\end{equation*}
where $\hat{Y}_m=2^{m-1}Y_m-2^{e}k$ is an integer.
So, the in-degree of $\Y$ in $\mathbb{F}_{\rm H,e}$ is equal to the number of integers in the above set, namely $2^{m-1}$.
\end{proof}

To facilitate the following discussion, we define a $2^{m-1}$-way ``full recursive tree" with height $i$ as $\mathcal{T}_{i}$:
1) link every element in a set containing $2^{m-1}$ nodes to a given node with a direct edge;
2) recursively link every node in the $j$-th level of the tree with $2^{m-1}$ new nodes for $j=1\sim i-1$ if $i\ge 2$ (every edge directs from the higher level to the lower level).
Referring to Property~\ref{prop:indgreem}, the basic unit of $\mathbb{F}_{\rm H, e}$ is a tree
containing one non-leaf node (the root) and its $2^{m-1}$ child nodes connected by directed edges.
Such a unit is isomorphic to $\mathcal{T}_1$.
When $e=3, m=3$, the $\mathcal{T}_1$ corresponding to HDBM~(\ref{eq:B_He}) is shown in Fig.~\ref{fig:3DBakermap8}a), which 
is a basic part of the subgraph shown in Fig.~\ref{fig:3DBakermap8}b).
The entire functional graph is depicted in Fig.~\ref{fig:3DBakermap8}c).
Observing Fig.~\ref{fig:3DBakermap8}, one can see that 
a functional graph consists of some connected components isomorphic to $\mathcal{T}_3$, which is composed of the basic part isomorphic to $\mathcal{T}_1$ in a semi-fractal way.
First, to analyze the structure of the $\mathbb{F}_{\rm H,e}$, it is necessary to divide $\mathbb{F}_{\rm H,e}$ into several subgraphs as described in Property~\ref{prop:H,e:Ftwopart}.
Then, based on Property~\ref{prop:H,e:SisoTi}, the structures of the tree and semi-fractal block are disclosed in Properties~\ref{prop:0tree} and~\ref{prop:0fractal}, respectively.
So, the entire functional graph is also disclosed in Proposition~\ref{prop:funMap}.

\begin{figure}[!htb]
 \centering
 \begin{minipage}{0.5\twofigwidth}
 \centering
 \vspace{2em}\includegraphics[width=0.5\twofigwidth]{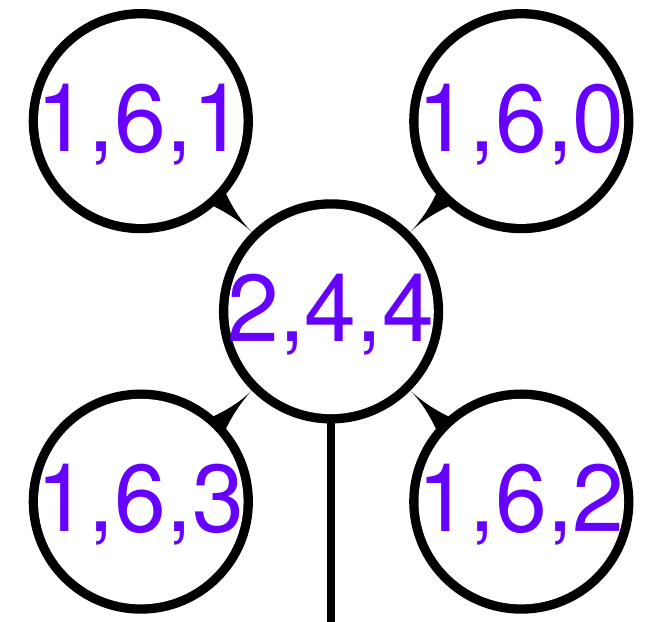}
 \vspace{0.5em}\subcaption*{a)}
 \end{minipage}
 \hspace{2em}
 \begin{minipage}{0.8\twofigwidth}
 \centering
 \includegraphics[width=0.8\twofigwidth]{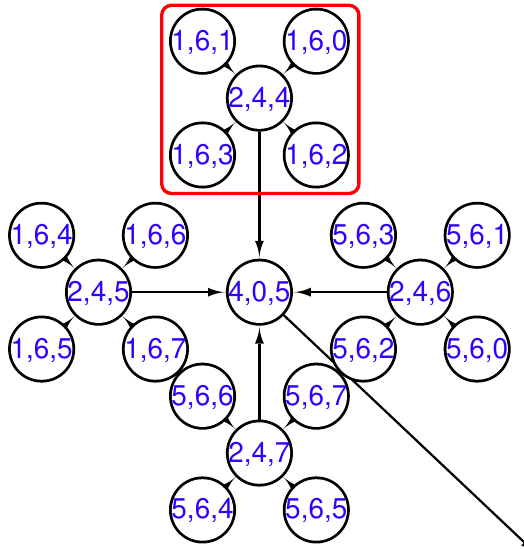}
 \subcaption*{b)}
 \end{minipage}
 \begin{minipage}{1.0\BigOneImW}
 \centering
 \includegraphics[width=1.0\BigOneImW]{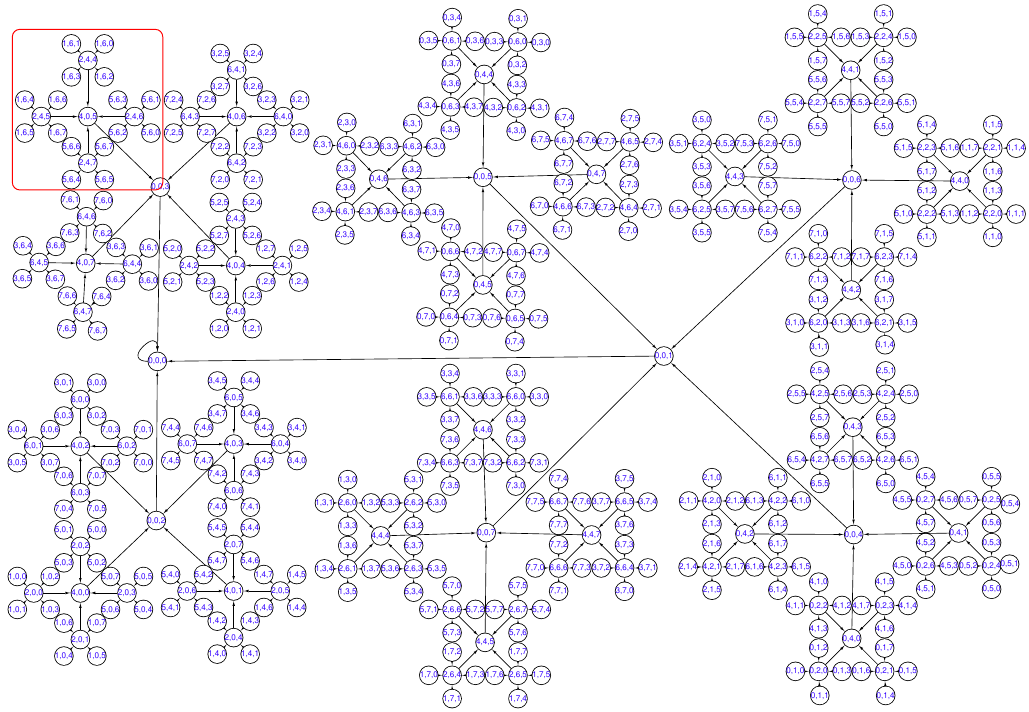}
 \subcaption*{c)}
 \end{minipage}
 \caption{The hierarchical structure of the functional graph of HDBM~(\ref{eq:B_He}) implemented in a fixed-point arithmetic domain with $m=3$, $e=3$:
 a) subgraph I;
 b) subgraph II;
 c) the entire functional graph.}
\label{fig:3DBakermap8}
\end{figure}

\begin{Property}\label{prop:H,e:Ftwopart}
The whole functional graph of HDBM~(\ref{eq:B_He}) is a connected component and 
composed of subgraph $\hat{\s}_0$ spanned by all non-zero elements in $\s_0=\{\X\ \vert\ \tilde{\X}={\bf 0}\}$ and other subgraphs rooted at every leaf node of $\hat{\s}_0$.
\end{Property}
\begin{proof}
As subgraph $\hat{\s}_0$ dominate the basic structure of the whole functional graph of HDBM~(\ref{eq:B_He}),
the evolution analysis of every node is divided into the following two cases:
\begin{itemize}
\item $\X \in \s_0$: 
one can get $\mathcal{B}_{\rm e}(\X)_i=0$ for 
$i=1\sim m-1$. Referring to Eq.~\eqref{eq:YXsub}, for any positive integer $j$, one has 
\begin{equation}
\label{eq:H,e:S0Bjs}
\tilde{\X}^{(j)}_{\rm e}={\bf 0},
\end{equation}
where ${\X}^{(j)}_{\rm e} = \mathcal{B}_{\rm e}^{j}(\X)$  and $\tilde{\X}^{(j)}_{\rm e}$ is the vector composed by the first $m-1$ elements of ${\X}^{(j)}_{\rm e}$. 
It follows from Eq.~\eqref{eq:H,e:Ym} that $X^{(1)}_{\rm e, m} =\mathrm{R}(X_m/2^{m-1})\leq X_m/2^{m-1}$, where $X^{(1)}_{\rm e, i}$ is the $i$-th element of $\mathcal{B}_{\rm e}(\X)$, then 
\begin{equation}
\label{eq:H,e:BjX}
{X}^{(j)}_{\rm e, m}\leq X_m/2^{j(m-1)}
\end{equation}
for any positive integer $j$.
So, there exists $n_1$ satisfying ${X}^{(n_1)}_m<1$, namely ${\X}^{(n_1)}={\bf 0}$.
So, all nodes in $\s_0$ form a subgraph $\hat{\s}_0$ rooted at node $\bf{0}$.
 
\item $\X \not\in \s_0$: it follows from Eq.~\eqref{eq:YXsub} that there exists a positive integer $n_2$ satisfying
\begin{equation}
\left\{
 \begin{split}
 \tilde{\X}^{(n_2-1)}_{\rm e}&\neq{\bf 0},\\
 \tilde{\X}^{(n_2)}_{\rm e} &={\bf 0}.
 \end{split}
\right.
\end{equation}
So, one has ${\X}^{(n_2-1)}_{\rm e}\not\in C'_0$. According to Eq.~\eqref{eq:H,e:Ym}, ${\X}^{(n_2)}_{\rm e} \in \{\X \mid X_m \geq 2^{e-m+1}, \X\in\s_0\}$.
Then, by Eqs.~\eqref{eq:H,e:S0Bjs} and~\eqref{eq:H,e:BjX}, the set only include all leaf nodes of $\hat{\s}_0$. And for any node $\Y'$ in the set, because of Eqs.~\eqref{eq:YXsub} and~\eqref{eq:H,e:Ym}, one has $k'=\mathrm{R}(Y'_m/2^{e-m+1})$ and there exists $\X'$ satisfying $\mathcal{B}_{\rm e}(\X')=\Y'$ when $\tilde{\X}'=2^{e-1}\tilde{\bf p}_{\rm k'}$ and $X'_m = 2^{m-1}Y'_m-2k'2^{e}$.
It means that any node in the set is not a leaf node of $\mathbb{F}_{\rm H,e}$.
Therefore for any leaf node of $\hat{\s}_0$, there is a subgraph is rooted at it.
\end{itemize}
Combining the above two cases, the property holds.
\end{proof}

\begin{Property}
\label{prop:H,e:SisoTi}
Given a node $\Y^*\neq {\bf 0}$ and positive integer $i^*$, if the cardinality of set 
\begin{equation}
\label{eq:setx*}
\{\X\ \vert\ \X^{(i)}_{\rm e}=\Y^*, i\in\{0, 1, \cdots, i^*\}\}
\end{equation}
is equal to $\sum_{i=0}^{i^*}2^{i(m-1)}$,
then all nodes in the above set compose a subgraph isomorphic to $\mathcal{T}_{i^*}$
with respect to their mapping relation determined by HDBM~(\ref{eq:B_He}),
where $\X^{(i)}$ represents the $i$-th iteration of the map.
\end{Property}
\begin{proof}
Referring to Property~\ref{prop:indgreem}, the cardinality of set~\eqref{eq:setx*} is less than or equal to $\sum_{i=0}^{i^*}2^{i(m-1)}$.
If the cardinality of set~\eqref{eq:setx*} is $\sum_{i=0}^{i^*}2^{i(m-1)}$, the tree rooted at $\Y^*$ has at least $i^*$ levels and the $i$-th level of the tree has $2^{i(m-1)}$ nodes. Otherwise, it would contradict with Property~\ref{prop:indgreem}.
So, the tree is isomorphic to $\mathcal{T}_{i^*}$. 
\end{proof}

\begin{Property}\label{prop:0tree}
Ignoring the self-link to root node numbered zero, $\hat{S}_0$ is isomorphic to a tree composing of root ${\bf 0}$, $2^{e-i_e(m-1)}-1$ subtrees isomorphic to
$\mathcal{T}_{i_e}$ and $2^{m-1}-2^{e-i_e(m-1)}$ subtrees 
isomorphic to
$\mathcal{T}_{i_e-1}$, where $i_e=\mathrm{R}(\frac{e}{m-1})$.
\end{Property}
\begin{proof}
Referring to Property~\ref{prop:H,e:Ftwopart}, all nodes in $\s_0$ form a subgraph rooted at node $\bf{0}$.
Divide $\s_0$ into subsets $\{\bf{0}\}$ and
\begin{equation}
\label{eq:3D:Sk}
\begin{split}
\s_{a, i} & = \{\X\ \vert\ \X\in \mathrm{\s_0}, \X^{(i)}_{\rm e}=(0, 0, \cdots, a)^\intercal\}\\
 & =\{\X\ \vert\ X_m \in[a\cdot 2^{i(m-1)}, (a+1)\cdot 2^{i(m-1)}),\\
& \quad \quad\quad\quad\X\in \mathrm{\s_0}, (a+1)\cdot 2^{i(m-1)} \leq 2^e\},
\end{split}
\end{equation}
where $a= 1, 2, \cdots, 2^{m-1}-1$, and $i= 0, 1, \cdots, i_e$.
Then, one has
\[
\bigcup_{i=0}^{i^*}\s_{a, i}= \{\X\ \vert\ \X^{(i)}_{\rm e}=(0, 0, \cdots, a)^\intercal, i\in\{0, 1, \cdots, i^*\}\}
\]
and $\left|\cup_{i=0}^{i^*}\s_{a, i}\right|=\sum_{i=0}^{i^*-1}|\s_{a, i}|=\sum_{i=0}^{i^*-1}2^{i(m - 1)}$. It follows from Property~\ref{prop:H,e:SisoTi} that all elements in $\left|\cup_{i=0}^{i^*}\s_{a, i}\right|$ compose a subgraph isomorphic to $\mathcal{T}_{i^*}$.
And as $(k+1)\cdot 2^{i(m-1)} \leq 2^e$ in Eq.~\eqref{eq:3D:Sk}, there are two ranges of $i$ with different scopes of $a$:
\begin{itemize}
\item $a\in[1, 2^{e - i_e(m - 1)})$: $i\in 
\{0, 1, \cdots, i_e\}$. So, all nodes in $\bigcup\limits{}_{i=0}^{i_e}\s_{a, i}$ compose a subgraph isomorphic to $\mathcal{T}_{i_e}$.
 
\item $a \in[2^{e - i_e(m - 1)}, 2^{m - 1})$: $i\in \{0, 1, \cdots, i_e-1\}$. So, all nodes in $\bigcup\limits{}_{i=0}^{i_e-1}\s_{a, i}$ compose a subgraph isomorphic to $\mathcal{T}_{i_e-1}$.
\end{itemize}
Combining the above two cases and $\X^{(i)}_{\rm e}=\bf{0}$ for any $\X=(0, 0, \cdots, a)^\intercal$, this property is proved.
\end{proof}

\begin{Property}
\label{prop:0fractal}
Any leaf-node of the subgraph spanned with all nodes in $\s_0$ 
in $\mathbb{F}_{\rm H, e}$ is the root of a subgraph isomorphic to $\mathcal{T}_e$.
\end{Property}
\begin{proof}
Referring to Eq.~\eqref{eq:B_He}, a node $\X$ is a non-leaf node if $\X$ satisfies that $X_m\cdot 2^{m-1}-k\cdot 2^{e}$ is an integer and $X_i/2$ is an integer with $i\neq m$.
It is obvious that $X_m\cdot 2^{m-1}-k\cdot 2^{e}$ is always an integer. 
Let $X^{(-1)}_{\rm e, i}$ be the $i$-th element of the preimage of $\X$.
 When $i\in\{1,\cdots,m-1\}$, for any leaf-node $\X$ of the tree spanned with all nodes in $\s_0$, 
$X^{(-j)}_i \equiv 0 \pmod{2^{e-j}}$ holds. So, $X^{(-j)}_{\rm e, i}/2$ 
is an integer for any $j\in\{0, 1, \cdots, e-1\}$. When $j=e$, there exists $i\neq m$ satisfying that $X^{(-j)}_{\rm e, i}/2$ is not an integer.
It means that $\X^{(-j)}_{\rm e, i}$ is non-leaf node when $j\in\{0,1,\cdots,e-1\}$ and $\X^{(-j)}_{\rm e}$ is leaf node.
Combining Property~\ref{prop:indgreem} and the definition of $\mathcal{T}_i$, one has $\X^{(e)}_{\rm e, i}$ is the root of a subgraph isomorphic to $\mathcal{T}_e$.
\end{proof}

\begin{Proposition}\label{prop:funMap}
Ignoring the self-link to root node numbered zero, the functional graph of HDBM~(\ref{eq:B_He}) is isomorphic to a tree composed of root ${\bf 0}$, $2^{e-i_e(m-1)}-1$ subtrees $\mathcal{T}_{i_e+e}$ and $2^{m-1}-2^{e-i_e(m-1)}$ subtrees $\mathcal{T}_{i_e+e-1}$, where $i_e=\mathrm{R}(\frac{e}{m-1})$.
\end{Proposition}
\begin{proof}
 Combining Properties~\ref{prop:0tree} and~\ref{prop:0fractal}, this proposition holds.
\end{proof}

As shown in Fig.~\ref{fig:3DBakermap8},
all non-zero nodes of $\s_0=\{(0, 0, 0)^\intercal, (0, 0, 1)^\intercal, \cdots, (0, 0, 7)^\intercal \}$ compose two nodes (it can be seen as a tree isomorphic to $\mathcal{T}_0$) and one tree isomorphic to $\mathcal{T}_1$, which agrees with Property~\ref{prop:0tree}.
Then, the leaf nodes of the tree composed of all nodes in $\s_0$ is the root of a tree isomorphic to $\mathcal{T}_3$, which is precisely what Property~\ref{prop:0fractal} describes.
And the functional graph $\mathbb{F}_{\rm H, 3}$ is isomorphic to a tree composing of one sub-trees isomorphic to $\mathcal{T}_4$, and two sub-trees isomorphic to $\mathcal{T}_3$ as Proposition~\ref{prop:funMap} depicts.

\subsection{The functional graph of HDBM in a floating-point arithmetic domain}

Following the IEEE 754 standard, the most common technical standard for floating-point arithmetic, a sequence of $n$ bits, $\{b(i)\}_{i=0}^{n-1}$, is divided into three parts: a sign, a signed exponent and a mantissa \cite{754-2019}.
Let the length of exponent be $l$, and the length of mantissa be $m_f$; then the trailing significand field digit string is a sequence of $m_f$ bits $(a_1a_2 \cdots a_{m_f})_2$.
A number in the floating-point arithmetic domain is interpreted as
\begin{equation*}
\begin{split}
&v= \\
 &\begin{cases}
 0   &  \mbox{if } p=0, o=0;\\
 (-1)^s\cdot\left(\sum\limits_{i=1}^{m_f}b_{l + i}\cdot 2^{-i} \right)\cdot 2^{2-2^{l-1}} 
     &   \mbox{if }p=0, o\neq0;\\
 (-1)^s\cdot \infty      & \mbox{if } p=2^{l}-1, o = 0;\\
 \text{``not a number"} & \mbox{if } p=2^{l}-1, o \neq 0;\\
 (-1)^s\cdot\left(1+\sum\limits_{i=1}^{m_f}b_{l + i}\cdot 2^{-i} \right)\cdot 2^{p-o} & \mbox{otherwise},
 \end{cases}
\end{split}
\end{equation*}
where $s=b_0$,
$p=\sum_{i=0}^{l-1}b_{i+1}\cdot 2^i$, $o=2^{l-1}-1$. 
As for the floating-point arithmetic domain with parameters $l$ and $m_f$, the domain and range of HDBM~(\ref{eq:BH}) are both discrete set
$\{{\bf x} \mid x_1=\frac{X_1}{2^e}, x_2=\frac{X_2}{2^e}, \cdots,
x_m=\frac{X_m}{2^e}, X_i\in\mathbb{Z}'\}$, where $\mathbb{Z}'=\{X' \mid X'=\mathrm{R_f}(X), X\in\mathbb{Z}_{2^{e_{\rm m}}}\}$, $e_{\rm m}=m_f+2^{l-1}-2$, 
\begin{equation}
\label{eq:float:Rf}
\mathrm{R_f}(X) =X -
\begin{cases}
 X \bmod 2^{e_x- m_f}
 & \mbox{if } e_x\in[m_f, e_{\rm m}) \\
 0
 & \mbox{if } e_x \in [0, m_f)
\end{cases}
\end{equation}
denoting the quantization function exerting in the floating-point arithmetic domain hereinafter,
where
$e_x=\max\{e \mid X\geq 2^e\}$.
Hence, HDBM~(\ref{eq:BH}) can be represented as
\begin{equation}
\label{eq:B_Hf}
\mathcal{B}_{\rm f}(\X)= \mathrm{R_f}(2^{e_{\rm m} }\cdot \mathcal{B}(\X/2^{e_{\rm m}}))
\end{equation}
when it is implemented in the floating-point arithmetic domain.
For any $\X$, as $X_i\in\mathbb{Z}'$, one has 
\begin{multline}
X_i= \\
\begin{cases}
\left(1+\sum\limits_{i=1}^{m_f}a_i\cdot 2^{-i}\right)\cdot 2^{e_{x_i}}
 & \mbox{if } e_{x_i}\in[m_f, e_{\rm m}); \\
 \left(\sum\limits_{i=1}^{m_f}a_i\cdot 2^{-i}\right)\cdot 2^{m_f}
 & \mbox{if } e_{x_i} \in [0, m_f),
\end{cases} 
\label{eq:X_if}
\end{multline}
where $a_i$ is the $i$-th most significant bit of the mantissa of $X_i$.
So, it follows from $\Y=\mathcal{B}_{\rm f}(\X)$ that the relation between $Y_i$ and $\{a_i\}_{i=1}^{m_f}$ corresponding to $X_i$
can be divided into three cases concerning $i$ and $m$:
\begin{itemize}
\item $i\neq m$: it follows from Eq.~\eqref{eq:B_Hf} that $Y_i=2 X_i - 2^{e_{\rm m}}\beta(k,i)$. Then, according to the definition of $\beta(k,i)$, one has
\begin{equation*}
\beta(k,i)=
\begin{cases}
 1 & \mbox{if } e_{x_i}=e_{\rm m} -1; \\
 0 & \mbox{otherwise}.
\end{cases}
\end{equation*}
It means from Eq.~\eqref{eq:X_if} that
\begin{equation}
\label{eq:Yif}
Y_i=
 \begin{cases}
 \left(\sum\limits_{i=1}^{m_f}a_i \cdot 2^{-i}\right)\cdot 2^{e_{\rm m}}
 & \mbox{if } e_{x_i}=e_{\rm m} -1;\\
 \left(1+\sum\limits_{i=1}^{m_f}a_i \cdot 2^{-i}\right)\cdot 2^{ e_{x_i}+1}
 &\mbox{if } e_{x_i}\in[m_f, e_{\rm m}-1);\\
 \left(\sum\limits_{i=1}^{m_f}a_i \cdot 2^{-i}\right)\cdot 2^{m_f+1}
 & \mbox{if } e_{x_i} \in [0, m_f).
\end{cases}
\end{equation}

\item $i=m$ and $m-1 \leq m_f$: one has
\begin{equation}\label{eq:Y_mf'}
 Y_m=\mathrm{R_f}(Y'_m+k\cdot 2^{e_{\rm m}-m+1})
\end{equation}
and
\begin{IEEEeqnarray*}{rCl}
\IEEEeqnarraymulticol{3}{l}{
Y'_m = \mathrm{R_f}(X_m / 2^{m - 1})}\nonumber\\* \quad 
& = & 
\begin{cases}
\left(1+\sum\limits_{i=1}^{\tilde{e}_{x_m}}a_i \cdot 2^{-i}\right) \cdot 2^{\hat{e}_{x_m}}
 & \mbox{if } e_{x_m}\in[m_f, e_{\rm m});\\
 \left(\sum\limits_{i=1}^{\hat{m}_f}a_i \cdot 2^{-i}\right)\cdot 2^{\hat{m}_f}
 & \mbox{if }e_{x_m} \in [0, m_f),
\end{cases}
\label{eq:Y'm}
\end{IEEEeqnarray*}
where $\tilde{e}_{x_m}=\min(\hat{e}_{x_m}, m_f)$, 
$\hat{e}_{x_m}=e_{x_m}-m+1$ and $\hat{m}_f=m_f-m+1$.
\item $i=m$ and $m-1 > m_f$: 
the analysis of $Y_m$ is more simple than that on the second case
and omitted here.
\end{itemize}

Let $\mathbb{F}_{\rm H,f}$ denote the functional graph of $\mathcal{B}_{\rm f}(\X)$. Similar to $\mathbb{F}_{\rm H, e}$, every trajectory on $\mathbb{F}_{\rm H, f}$ satisfies Property~\ref{prop:float_0}. 
And this property is more obvious in Fig.~\ref{fig:float32}, which is the functional graph when $m=3$, $l= 3$, $m_f=1$. 
The in-degree of a non-leaf node $\Y$ in $\mathbb{F}_{\rm H, f}$ is decided by the scope of $Y_m$ as shown in Properties~\ref{prop:dYm2} and~\ref{prop:dYm1g}.

\begin{Property}
\label{prop:float_0}
Any trajectory in $\mathbb{F}_{\rm H, f}$ eventually evolve into fixed-point ${\bf 0}$.
\end{Property}
\begin{proof}
From Eqs.~\eqref{eq:X_if} and~\eqref{eq:Yif}, one can get $e_{y_i} = e_{x_i} + 1$ when $i\neq m$ and $e_{x_i}\neq e_{\rm m}-1$; the mantissa of $Y_i$ is $a_{j+1}\cdots a_{m_f}\underbrace{0000}_j,$ 
when $i\neq m$ and $e_{x_i}= e_{\rm m}-1$, where $j=\min\{i \mid a_i=1\}$.
So, there exists a positive integer $n_1$ satisfying
$\tilde{\X}^{(n_1)}_{\rm f}={\bf 0}$, where $\X^{(n_1)}_{\rm f}=\mathcal{B}_{\rm f}^{n_1}(\X)$.
When $n \geq n_1$, one can get $\tilde{\X}^{(n)}_{\rm f}= {\bf 0}$, then there exists a positive number $n_2$ satisfying $ {\X}^{(n_2)}_{\rm f} ={\bf 0}$ by Eq.~\eqref{eq:B_Hf}.
The trajectory in $\mathbb{F}_{\rm H, f}$ starting from initial point $\X$ eventually enters fixed-point ${\bf 0}$.
\end{proof}

As any node $\Y$ in $\mathbb{F}_{\rm H, f}$ represents an interval, the property of $\mathbb{F}_{\rm H, f}$ can be analyzed by interval arithmetic as \cite{Nathalie:NumericalRe:TC2014}.
And the unit in the floating-point arithmetic domain corresponding to $\Y$ in $\mathcal{B}_{\rm f}$ can be defined as Definition~\ref{def:cellf} by interval arithmetic.
\begin{Definition}
\label{def:cellf}
Define $\prod \limits_{i=1}^m\mathbb{Y}_{i}$ as a cube $\mathbb{Y}_{\rm f}$ for any given ${\bf Y}$, where \begin{equation}
\label{eq:float:IYi}
\mathbb{Y}_{i}=
\begin{cases}
 [Y_i, Y_i+2^{e_{y_i}-m_f}) & \mbox{if } e_{y_i}> m_f;\\
 [Y_i, Y_i+1)               & \mbox{if } e_{y_i}\leq m_f.
\end{cases}
\end{equation}
\end{Definition}
Let $\BC = \mathcal{B}_{\rm f}^{-1}(\mathbb{Y})$ with the interval algorithm.
Then, because $\mathrm{R_f}\left(\inf(\mathbb{Y}_i)\right)=\inf(\mathbb{Y}_i)$ and $\mathrm{R_f}\left(\sup(\mathbb{Y}_i)\right)=\sup(\mathbb{Y}_i)$, when $i \neq m$, one has
\begin{multline}
\label{eq:float:I'Y_i}
 \BC_i=
 \left[\frac{Y_i+\beta(k,i) 2 ^ {e_{\rm m}} }{2}, \frac{Y_i+\beta(k,i) 2 ^ {e_{\rm m}}+|\mathbb{Y}_{i}| }{2} \right);
\end{multline}
when $i = m$, one has
\begin{multline}
\label{eq:float:I'Y_m}
\BC_m=[2^{m-1}Y_m - k\cdot 2^{e_{\rm m}},\\
2^{m-1}(Y_m+|\mathbb{Y}_{m}|) - k\cdot 2^{e_{\rm m}}).
\end{multline}
The in-degree of a node $\Y$ in $\mathbb{F}_{\rm H, f}$ satisfies
$$d_{\bf Y}=|\{\X \mid \X \in \BC, \X\in \mathbb{Z}'\}|.$$
And the in-degree of $Y_i$ is
\begin{equation}
\label{eq:dYiXif}
d_{y_{i}}=
 |\{X_i \mid X_i\in \BC_i, X_i\in\mathbb{Z}'\}|
\end{equation}
in $\mathbb{F}_{\rm H, f}$.
According to HDBM~(\ref{eq:BH}), for any node $\Y$, there is $k$ such that $\Y\in C'_k$. So $d_{y_{i}}$ is independent for any $i$. It means that the in-degree of $\Y$ satisfies 
\begin{equation}\label{eq:f:dYs}
 d_{\Y}=\prod_{i=1}^{m}d_{y_i}.
\end{equation}
Then, node $\Y$ is the leaf node of $\mathbb{F}_{\rm H, f}$ if there exists $i$ satisfying $d_{y_{i}} = 0$.
Property~\ref{prop:dYi} presents the relation between the in-degree of node $\Y$ and that of $Y_m$ if $\Y$ is a non-leaf node.

\begin{Property}
\label{prop:dYi}
In $\mathbb{F}_{\rm H, f}$, the in-degree of any non-leaf node $\Y$ is equal to that of $Y_m$, namely $d_{\Y}=d_{y_m}$.
\end{Property}
\begin{proof}
 For any non-leaf node $\Y$, there exists 
$k$ satisfying $\Y\in D'_k$. 
When $i\neq m$, the analysis of $d_{y_i}$ is divided into two case with different $\beta(k,i)$:
\begin{itemize}
\item $\beta(k,i)=0$: there is $\BC_i=\left[\frac{Y_i }{2}, \frac{Y_i+ |\mathbb{Y}_{i}|}{2} \right)$ by Eq.~\eqref{eq:float:I'Y_i}.
Then, referring to $\Y$ is a non-leaf node, $Y_i$ is an even and $X_i=\frac{Y_i }{2}$.
It follows from Eq.~\eqref{eq:float:IYi} that $\BC_i= 1$.
So, $d_{y_i}= 1$ by Eq.~\eqref{eq:dYiXif} in $\mathbb{F}_{\rm H, f}$.
\item $\beta(k,i)=1$: there is
$$\BC_i=\left[\frac{Y_i+2 ^ {e_{\rm m}}}{2}, \frac{Y_i+ 2 ^ {e_{\rm m}}+|\mathbb{Y}_{i}|}{2} \right)$$
by Eq.~\eqref{eq:float:I'Y_i}.
As $\Y$ is a non-leaf node, $Y_i$ is an even and $X_i=\frac{Y_i+2 ^ {e_{\rm m}}}{2}$.
It follows from Eq.~\eqref{eq:float:IYi} that $\BC_i \subset \mathbb{X}_{i}$. So, $d_{y_i}= 1$ by Eq.~\eqref{eq:B_He}.
\end{itemize}
So, it follows from Eq.~\eqref{eq:f:dYs} that
\begin{equation}\label{eq:df}
 d_{\Y}=\prod_{i=1}^{m}d_{y_i}=d_{y_m}
\end{equation}
when $i\neq m$.
\end{proof}

\begin{figure}[!htb]
 \centering
 \begin{minipage}{1.0\BigOneImW}
 \includegraphics[width=1.0\BigOneImW]{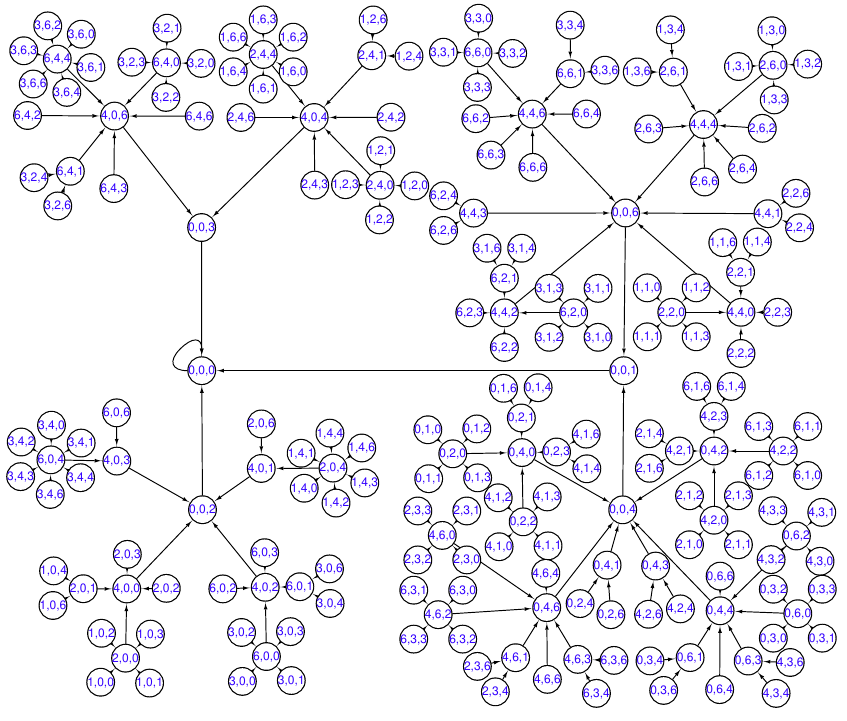}
 \end{minipage}
 \caption{The functional graph based on $\mathcal{B}_{\rm f}$ with $m=3$, $l= 3$, $m_f=1$.}
\label{fig:float32}
\end{figure}

\begin{Property}
\label{prop:dYm2}
When $\Y\in D'_0$ is a non-leaf node in $\mathbb{F}_{\rm H, f}$, its in-degree is
\begin{equation}
\label{eq:f:dXm}
 d_{\Y}= \\
 \begin{cases}
 1 & \mbox{if } e_{y_m} \in [m_f, e_{\rm m} - m +1); \\
 2^{e_{y_m}-\hat{m}_f} & \mbox{if } e_{y_m} \in [\hat{m}_f + 1, m_f);\\
 2^{m-1} & \mbox{if } e_{y_m} < \hat{m}_f+1,
 \end{cases}
\end{equation}
where $e_{y_m} = \max\{e \mid Y_m\geq 2^e\}$ and $\hat{m}_f=m_f-m+1$.
\end{Property}
\begin{proof}
When $\Y\in D'_0$, referring to Eqs.~\eqref{eq:B_Hf} and~\eqref{eq:float:IYi}, one has $e_{y_m}\leq e_{\rm m}-m+1$ and
\begin{equation}
    \label{eq:D0dYm}
 \BC_m= [2^{m-1}Y_m, 2^{m-1}Y_m+2^{m-1}|\mathbb{Y}_{m}|),
\end{equation}
where 
\begin{equation}
\label{eq:|Y|}
|\mathbb{Y}_{m}|=
 \begin{cases}
 2^{e_{y_m}-m_f} &  \mbox{if } e_{y_m}> m_f;\\
 1               &  \mbox{if } e_{y_m}\leq m_f.
\end{cases}
\end{equation}
Because $Y_m \in \mathbb{Z}'$, one has 
\begin{equation}\label{eq:Ymmax}
 Y_m\leq 2^{e_{y_m}+1} - |\mathbb{Y}_{m}|.
\end{equation}
And it follows from Eq.~\eqref{eq:Y_mf'} that $X_m = 2^{m-1}Y_m$ when $\Y\in D'_0$. Then it follows from Eq.~\eqref{eq:float:IYi} that 
\begin{equation*}
|\mathbb{X}_{m}| =
    \begin{cases}
     2^{e_{y_m}- \hat{m}_f}   &  \mbox{if } e_{y_m}\in [m_f-m+2, {e_{\rm e}- m});\\
     1   &  \mbox{if } e_{y_m} <{m_f-m+2},
    \end{cases}
\end{equation*}
for any $X_m\in\BC_m$. Then, it follows from Eq.~\eqref{eq:dYiXif} that 
\begin{equation}
 \label{eq:float:H:dyi}
 d_{y_m} = 
 \begin{cases}
 \frac{|\BC_m|}{2^{e_{y_m}-\hat{m}_f}} &\mbox{if } e_{y_m}\in [\hat{m}_f+1, {e_{\rm m}- m + 1});\\
 |\BC_m| &\mbox{if } e_{y_m} \in [0,{\hat{m}_f+1}),
 \end{cases}
\end{equation}
where $e_{x_m} = \max\{e \mid X_m\geq 2^e\}$.
So, according to Property~\ref{prop:dYm2}, above equation and Eq.~\eqref{prop:dYm2}, Eq.~\eqref{eq:f:dXm} holds.
\end{proof}

\begin{Property}
\label{prop:dYm1g}
When $\Y\in D'_k(k\neq 0)$ is a non-leaf node in $\mathbb{F}_{\rm H, f}$, its in-degree is
\begin{equation}\label{eq:float:dYm}
 d_{\Y}=
 \begin{cases}
  2^{{i}_k + e_m - e_{x'_m}}
 &\mbox{if } X'_m \neq 0,\\
2^{m_f}(i_k + e_m - 2m_f +1) 
 & \mbox{if } X'_m = 0, 
 \end{cases}
\end{equation}
where $i_k=\max\{e \mid k\geq 2^e\}$, $e_{x'_m}=\max\{e \mid X'_m\geq 2^e\}$ and $X'_m= 2^{m-1}Y_m - k\cdot 2^{e_{\rm m}}$.
\end{Property}
\begin{proof}
When $\Y\in D'_k(k\neq 0)$, referring to Eq.~\eqref{eq:Y_mf'}, one can get $e_{y'_m} < e_m- m+1$, that is, 
\begin{equation}
\label{eq:3:eymk}
e_{y_m}=i_k + e_m- m+1.
\end{equation} Then, it follows from Eqs.~\eqref{eq:float:IYi} and~\eqref{eq:float:I'Y_m} that 
\begin{equation*}
 |\mathbb{Y}_{m}|=
 2^{i_k + e_m - m_f - m+1}
\end{equation*}
and
\begin{multline*}
\BC_m=
 [2^{m-1}Y_m - k\cdot 2^{e_{\rm m}},\\ 2^{m-1}Y_m+2^{i_k + e_m - m_f} - k\cdot 2^{e_{\rm m}}).
\end{multline*}
Referring to Eq.~\eqref{eq:3:eymk} and $Y_m = 0\pmod{|\mathbb{Y}_{m}|}$, one can get $X'_m = 0\pmod{2^{i_k + e_m - m_f}}$. Let $X'_m = C2^{i_k + e_m - m_f}$, one has
\begin{equation}\label{eq:DkdYm}
\BC_m = [C2^{i_k + e_m - m_f}, (C+1)2^{i_k + e_m - m_f}),
\end{equation}
where $C$ is a constant.
So, one can deduce the value of $d_{\Y}$ by two cases of $X'_m$:
\begin{itemize}
 \item $X'_m \neq 0$: it means $C\neq 0$ and $e_{x'_m} \geq i_k + e_m - m_f$. According to the design of floating-point domain, one has $2^{l-1} -2 \geq m_f$, then, it follows 
 \begin{equation}\label{eq:3:em>mf}
 e_m - m_f \geq m_f.
 \end{equation}
 Referring to Eq.~\eqref{eq:DkdYm}, one has
 $$\BC_m \subset[2^{e_{x'_m}}, 2^{e_{x'_m}+1}).$$ 
 It follows from Eq.~\eqref{eq:float:H:dyi} that $d_{\Y} = 2^{i_k + e_m - e_{x'_m}}$.
 
 \item $X'_m = 0$: according to Eq.~\eqref{eq:DkdYm}, one has $\BC_m = [0, 2^{i_k + e_m - m_f})$. Referring to Eq.~\eqref{eq:3:em>mf}, one can get 
 $$\BC=[0, 2^{m_f})\bigcup\left(\bigcup_{i=1}^{i_k + e_m - 2m_f}[2^{m_f+i}, 2^{m_f+i+1})\right)$$
 and it follows from Eq.~\eqref{eq:float:H:dyi} that $d_{\Y}=2^{m_f} \cdot (i_k + e_m - 2m_f+1)$. 
\end{itemize}
Combining the above analysis, Eq.~\eqref{eq:float:dYm} holds. 
\end{proof}

\subsection{Relation between functional graphs in two arithmetic domain}

For a floating-point arithmetic domain, the length of the minimum interval, which is represented by a node, is $\frac{1}{2^{e_{m}}}$. So for a Lipschitz-
continuous function $f:[0,1)\rightarrow [0, 1)$, all nodes in the functional graph $\mathbb{F}_{\rm f}$ with respect of $f$ also exist in the $\mathbb{F}_{\rm e_m}$. 
For any node $Y_{\rm f}$ in the $\mathbb{F}_{\rm f}$, 
there are $|\mathbb{Y}_{\rm f}|$ nodes $\{Y_{\rm e, i}\}_{i=1}^{|\mathbb{Y}_{\rm f}|}$ which is in the $\mathbb{Y}_{\rm f}$. 
Referring to Eq.~\eqref{eq:float:IYi}, one has $\mathbb{Y}_{\rm f} = \bigcup\mathbb{Y}_{\rm e, i}$ and $f^{-1}(\mathbb{Y}_{\rm f}) = \bigcup f^{-1}(\mathbb{Y}_{\rm e, i})$.
For any node $Y_{\rm e, i}$, the in-degree in $\mathbb{F}_{\rm e_m}$ is similar to Eq.~\eqref{eq:dYiXif} that 
\begin{equation}\label{eq:3cde}
\begin{split}
d_{\rm e, i} = &  |\{X \mid X\in f^{-1}(\mathbb{Y}_{\rm e, i}), X\in \mathbf{Z}\}|\\
             = & \CEIL{|f^{-1}(\mathbb{Y}_{\rm e, i})|}.
\end{split}
\end{equation}
Then, the functional graph $\mathbb{F}_{\rm f}$ 
can deduce from $\mathbb{F}_{\rm e_m}$ by a method with two steps as follows: traversing the nodes in $\mathbb{F}_{\rm e_m}$, if node $Y\not\in\mathbf{Z}'$, remove the node and link all preimages of $Y$ to $Y_{\rm f}=\mathrm{R}_{\rm f}(Y)$.

For any $Y_{\rm f}$, similarly to Eq.~\eqref{eq:float:H:dyi}, the in-degree $d_f$ of the node $Y_{\rm f}$ can be expressed by
\begin{equation}\label{eq:3cdf}
d_{f} = 
 \begin{cases}
 \CEIL{\frac{|f^{-1}(\mathbb{Y}_{\rm f})|}{2^{e- m_f}}} &\mbox{if }|f^{-1}(\mathbb{Y}_{\rm f})| \subset [2^e,2^{e+1});\\
 \CEIL{|f^{-1}(\mathbb{Y}_{\rm f})|} &\mbox{if } |f^{-1}(\mathbb{Y}_{\rm f})| \subset [0,2^{m_f}),
 \end{cases}
\end{equation}

\begin{Property}\label{prop:simife}
When $\mathbb{Y}_{\rm f}$ and $f^{-1}(\mathbb{Y}_{\rm f})$ are both in $[0, 2^{m_f})$, the tree composing of a root $Y_{\rm f}$ and its preimage in $\mathbb{F}_{\rm f}$ is isomorphic to the tree composing of root $Y_{\rm f}$ and its preimage in $\mathbb{F}_{\rm e_m}$.
\end{Property}
\begin{proof}
There is $|\mathbb{Y}_{\rm f}| = |\mathbb{Y}_{\rm e}| = 1$. It follows from Eqs.~\eqref{eq:3cde} and ~\eqref{eq:3cdf} that $d_f = \lceil\frac{|f^{-1}(\mathbb{Y}_{\rm f})|}{|\mathbb{Y}_{\rm f}|}\rceil=\lceil|f^{-1}(\mathbb{Y}_{\rm e})|\rceil$ and $d_f= d_{e, 1}$. So, this Property holds.
\end{proof}

When $\mathbb{Y}_{\rm f}$ or $f^{-1}(\mathbb{Y}_{\rm f})$ are not in $[0, 2^{m_f})$, it is possible to calculate $|f^{-1}(\mathbb{Y}_{\rm f})|$ by using a Taylor expansion of $f$. 
For $\mathcal{B}_{\rm e}$ and $\mathcal{B}_{\rm f}$, the order of its Taylor expansion is 1.
According to Properties~\ref{prop:indgreem},~\ref{prop:dYm2} and~\ref{prop:dYm1g}, the ratio of the in-degree of a node $Y_{\rm f}$ in the $\mathbb{F}_{\rm f}$ with it in the $\mathbb{F}_{\rm e_m}$ can be calculated. 
And the disturbance of a given $\x$ can be deduced from Property~\ref{prop:Bf-i=Xm}, as shown in Properties~\ref{prop:errorFI0} and~\ref{prop:errorFI}.

\begin{Property}
\label{prop:Bf-i=Xm}
For any node $\X$ in $\mathbb{F}_{\rm H, f}$, there exists $\X'$ in $\mathbb{F}_{\rm H, e_{\rm m}}$ and $\mathcal{B}_{\rm f}(\X)-\mathcal{B}_{\rm e}(\X')=\Delta_m,$ where $\Delta_m= \mathcal{B}_{\rm f}(\X)_m-\mathcal{B}_{\rm e}(\X')_m$.
\end{Property}
\begin{proof}
First, it follows from Eqs.~(\ref{eq:B_He}) and~\eqref{eq:B_Hf} that $\mathcal{B}_{\rm f}(\X)_{\rm s}=\mathcal{B}_{\rm e}(\X')_{\rm s}$.
Then, according to the definition of $C'_k$, if $\X\in C'_k$ in $\mathbb{F}_{\rm H, f}$, one can get $\X\in C'_k$ in $\mathbb{F}_{\rm H, e_{\rm m}}$. So, one has $\mathcal{B}_{\rm f}(\X)-\mathcal{B}_{\rm e}(\X')=\Delta_m$.
\end{proof}

\begin{Property}
\label{prop:errorFI0}
For any node $\X\in C'_0$ in $\mathbb{F}_{\rm H, f}$, there also exists $\X$ in $\mathbb{F}_{\rm H, e_{\rm m}}$ and 
\begin{equation}
\mathcal{B}_{\rm f}(\X)-\mathcal{B}_{\rm e}(\X)=0.
\end{equation} 
\end{Property}
\begin{proof}
When $\X \in C'_0$, it follows from Eqs.~\eqref{eq:X_if},~\eqref{eq:Y_mf'} and~\eqref{eq:Y'm} that $\mathcal{B}_{\rm e}(\X)_m = \mathrm{R}(X_m/2^{m-1})$ and $\mathcal{B}_{\rm f}(\X)_m = \mathrm{R}_{\rm f}(X_m/2^{m-1})$. 
Referring to Eq.~\eqref{eq:float:Rf} and $X_m\in \mathbb{Z}'$, one has $\mathrm{R}_{\rm f}(X_m/2^{m-1}) = \mathrm{R}(X_m/2^{m-1})$. 
Then, $\Delta_m=\mathcal{B}_{\rm f}(\X)_m-\mathcal{B}_{\rm e}(\X)_m=0$ by Eq.~\eqref{eq:Y_mf'}.
According to Property~\ref{prop:Bf-i=Xm}, one has $\mathcal{B}_{\rm f}(\X)-\mathcal{B}_{\rm e}(\X)=0$.
\end{proof}

\begin{Property}
\label{prop:errorFI}
For any node $\X=(X_1, X_2, \cdots, X_m)^\intercal\in C'_k$ in $\mathbb{F}_{\rm H, f}$, there also exists $\X$ in $\mathbb{F}_{\rm H, e_{\rm m}}$ and
\begin{multline}\label{eq:errorFI}
\mathcal{B}_{\rm f}(\X)-\mathcal{B}_{\rm e}(\X) =\\
 \begin{cases}
 \sum\limits_{i=e_{x_m}-\hat{i}_k}^{m_f}a_{i}\cdot 2^{\hat{e}-i}
 & \mbox{if } e_{x_m}\in [\hat{i}_k, e_{\rm m}-1);\\
 2^{m - 1} X_m & \mbox{if } e_{x_m} \in [0, \hat{i}_k),
 \end{cases} 
\end{multline}
where $\hat{i}_k=i_k+e_{\rm m}-m_f$, $k \in \{1, 2, \cdots, 2^{m-1}-1\}$, $i_k= \max\{i \mid k\geq 2^{i}\}$.
\end{Property}
\begin{proof}
First, from Eq.~\eqref{eq:Y_mf'}, one can get
$\mathcal{B}_{\rm f}(\X)_m=\mathrm{R}_{\rm f}\left(Y'_m+k\cdot 2^{e_{\rm m}-m+1}\right)$.
According to Eq.~\eqref{eq:Y'm}, one can get $Y'_m < 2^{e_{\rm m}-m+1}$, namely 
\begin{equation}\label{eq:float:e_ym}
\begin{split}
 e_{y_m}&=\max(e_{y'_m}, e_{\rm m}+i_k-m_f - m+1)\\
 &=\hat{i}_k+m_f - m+1
\end{split}
\end{equation} and
\begin{equation}\label{eq:BHFm}
 \mathcal{B}_{\rm f}(\X)_m\in[k\cdot 2^{e_{\rm m}-m+1}, (k+1)\cdot 2^{e_{\rm m}-m+1}).
\end{equation}
So, it follows from Eqs.~\eqref{eq:Y'm} and~\eqref{eq:BHFm} that
\begin{multline*}
\begin{split}
 & \mathcal{B}_{\rm f}(\X)_m=\\
 & \left(1+\sum_{i=1}^{e_{y'_m}-\hat{i}_k+m-1}a_i\cdot 2^{-i}\right)\cdot 2^{e_{y'_m}}+k\cdot 2^{e_{\rm m}-m+1}.
\end{split}
\end{multline*}
Then, because of Eqs.~\eqref{eq:X_if} and~\eqref{eq:Y'm}, one has $\mathcal{B}_{\rm e}(\X)_m=\mathrm{R}(X_m/2^{m-1})=Y'_m+k\cdot 2^{e_{\rm m}-m+1}$.
So, one can get $\Delta_m$ by $e_{x_m}$ from two cases:
\begin{itemize}
\item $e_{x_m} \in [\hat{i}_k, e_{\rm m})$: one has $e_{y'_m}=\hat{e}_m$. So, 
 \begin{equation*}
 \Delta_m=\sum_{i=e_{x_m}-\hat{i}_k}^{m_f}a_{i}\cdot 2^{\hat{e}-i}.
 \end{equation*}
 \item $e_{x_m} \in [0, \hat{i}_k)$: one has $e_{y'_m}<\hat{i}_k-m+1$, which means $Y_m=k\cdot 2^{e_{\rm m}-m+1}$ and $\Delta_m= Y'_m=2^{m - 1} X_m$.
\end{itemize}
From the above, it can be seen that Eq.~\eqref{eq:errorFI} holds by 
Property~\ref{prop:Bf-i=Xm}.
\end{proof}

\section{Conclusion}

This paper analyzed the structure of the two-dimensional generalized baker's map and 
its higher-dimensional version in a digital domain. 
Then, as $e$ increases, the upper bound
of in-degree of the node in the structure of the generalized baker's map in a digital arithmetic domain was obtained. The regular patterns of the phase space of baker's map implemented on a digital computer were reported, which is dramatically different from that in an infinite-precision domain.
An invariable maximum in-degree exists, no matter what the implementation precision is. The in-degree distribution approaches a constant as precision increases.  The functional graph in a fixed-precision arithmetic domain exhibits self-similarity, discernible through fractal dimensions, with a specific case evolving into a fractal pattern as precision augments. Comparatively, the properties of the functional graph in the floating-point arithmetic domain diverge from those in the fixed-point arithmetic domain. Nevertheless, the characteristics of specific local graphs share similarities between the two domains. This analysis method can be extended to the variants of baker's map and other chaotic maps and promote some research based on the floating-point arithmetic domain, such as numerical reproducibility and accuracy loss.

\bibliographystyle{IEEEtran_doi}
\bibliography{Baker}
\graphicspath{{author_figures_pdf/}}

\begin{IEEEbiography}[{\includegraphics[width=1.1in, height=1.25in,clip,keepaspectratio]{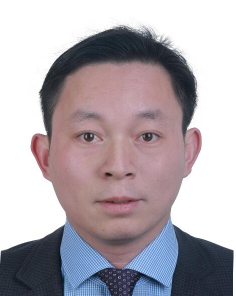}}]{Chengqing Li} (M'07-SM'13) obtained his M.Sc. degree in Applied Mathematics from Zhejiang University, China, in 2005, followed by a Ph.D. in Electronic Engineering from the City University of Hong Kong in 2008. He served as a Post-Doctoral Fellow at The Hong Kong Polytechnic University until September 2010. Subsequently, he joined the College of Information Engineering at Xiangtan University, China. From April 2013 to July 2014, he worked at the University of Konstanz, Germany, supported by the Alexander von Humboldt Foundation. Since April 2018, he has been a Professor at the School of Computer Science and Electronic Engineering, Hunan University, China, and currently holds the position of Full Professor at the School of Computer Science, Xiangtan University, China. His research primarily focuses on the dynamics analysis of digital chaotic systems and their applications in multimedia security. Over the past 21 years, he has published over 70 papers in this area, receiving more than 6100 citations and achieving an h-index of 41. 
He is a Fellow of the IET.
\end{IEEEbiography}

\vskip 0pt plus -1fil

\begin{IEEEbiography}[{\includegraphics[width=1in,height=1.25in,clip,keepaspectratio]{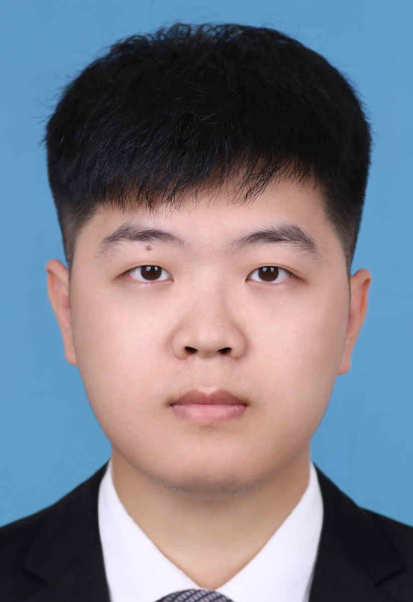}}]{Kai Tan}
received B.Sc. degree in mechanism design, manufacturing and automatization at the School of Mechanical Engineering, Xiangtan University in 2015. He received his M.Sc. degree in computer science at the School of Computer Science, Xiangtan University in 2020.
Now, he is pursing PhD degree at the same school.
His research interests include complex networks and nonlinear dynamics.
\end{IEEEbiography}
\end{document}